\documentclass[12pt]{amsart}
\usepackage[top=1.1in, hmargin=0.8in,height=8.6in]{geometry}
\usepackage{amssymb,amsthm, times}
\usepackage{delarray,verbatim}
\usepackage{multirow}
\usepackage{ifpdf}
\ifpdf
\usepackage[pdftex]{graphicx}
 \else
\usepackage[dvips]{graphicx}
 \fi

\usepackage{bm}

\usepackage{ifpdf}
\usepackage{color}
\definecolor{webgreen}{rgb}{0,.5,0}
\definecolor{webbrown}{rgb}{.8,0,0}
\definecolor{emphcolor}{rgb}{0.95,0.95,0.95}

\usepackage{hyperref}
\hypersetup{%
          colorlinks=true,
          linkcolor=webbrown,
          filecolor=webbrown,
          citecolor=webgreen,
          breaklinks=true}
\ifpdf \hypersetup{pdftex,
            pdfstartview=FitH, %%Fit, FitB, FitH
            bookmarksopen=true,
            bookmarksnumbered=true
} \else \hypersetup{dvips} \fi

\newcommand{\lapinv}{\Phi(q)}

\numberwithin{equation}{section}

\newtheorem{proposition}{Proposition}[section]

\newtheorem{remark}{Remark}[section]
\newtheorem{lemma}{Lemma}[section]

\newtheorem{assump}{Assumption}[section]

\numberwithin{remark}{section} \numberwithin{proposition}{section}
\numberwithin{corollary}{section}
\newcommand {\R}{\mathbb{R}}

\newcommand {\F}{\mathcal{F}}

\newcommand {\N}{\mathbb{N}}
\newcommand {\p}{\mathbb{P}}

\newcommand {\E}{\mathbb{E}}

\newcommand {\M}{\mathcal{M}}

\newcommand{\diff}{{\rm d}}

\newcommand{\lev}{L\'{e}vy }

\newcommand{\ex}{\mathbb{E}}
\newcommand{\setT}{\mathcal{T}}
\newcommand{\be}{\begin{equation}}
\newcommand{\ind}{1\hspace{-2.1mm}{1}} %Indicator Function
\newcommand{\ee}{\end{equation}}

\title[Optimal Multiple Stopping, Canadization, and Phase-Type Fitting]{An analytic  recursive method for optimal multiple stopping: Canadization and
phase-type fitting }
\thanks{This version: \today. }
\author[T. Leung]{Tim Leung}
\address[T. Leung]{IEOR Department, Columbia University, New York, USA.}
\email{leung@ieor.columbia.edu}
\author[K. Yamazaki]{Kazutoshi Yamazaki}
\address[K. Yamazaki]{Department of Mathematics,
Kansai University, Osaka, Japan }
\email{kyamazak@kansai-u.ac.jp}
\author[H. Zhang]{Hongzhong Zhang}
\address[H. Zhang]{Statistics Department, Columbia University, New York, USA}
\email{hzhang@stat.columbia.edu}

\date{}
\begin{document}

\begin{abstract}
We study an optimal multiple stopping problem  {for call-type payoff} driven by a spectrally negative \lev process.   The stopping times are separated by constant refraction times, and the discount rate can be positive or negative.  The computation involves a distribution of the \lev  process at a constant horizon and hence the solutions in general cannot be attained analytically.  Motivated by the  maturity randomization (Canadization) technique by Carr \cite{Carr_1998}, we approximate  the refraction times by independent, identically distributed \ Erlang random variables. In addition, fitting random  jumps  to  phase-type distributions, our method involves  repeated  integrations with respect to the \emph{resolvent measure} written in terms of the  {scale function} of the underlying  \lev process.  We derive a recursive    algorithm to compute the value function in closed form, and sequentially determine the optimal exercise thresholds.  A series of numerical examples are provided to compare   our analytic formula to results from Monte Carlo simulation.
\\

  \noindent \textbf{JEL Classification:} G32, D81, C61 \\
\noindent \textbf{Mathematics Subject Classification (2010):}  60G40, 60J75, 65C50 \\
\noindent \small{\textbf{Keywords:} optimal multiple stopping, refraction times,  maturity randomization, phase-type fitting,  \lev processes}   \\

\vspace{10pt}
%\textcolor[rgb]{1.00,1.00,1.00}{.}\\

\end{abstract}

\maketitle

\section{Introduction}
A wide array    of  financial applications can be formulated as  optimal multiple stopping problems.     These include energy delivery contracts such as swing options \cite{carmonaDayanik08,Carmona2008,ZeghalSwing},  derivatives liquidation \cite{Henderson2011,Leung2011b,Leung2011c},   real option analysis  \cite{Chiara,EricTim13,DixitPindyck,McDonald1985}, as well as employee stock options \cite{GrasselliHenderson2008,LeungSircarESO_MF09,leungsircar2} potentially  with additional reload and  shout options \cite{kwokshout}. In many of these applications,    consecutive stopping times are  separated by a constant or random  period. In the literature, especially that of swing options, this timing constraint is commonly referred to as the \emph{refraction period}.  In real option analysis, the refraction period can be interpreted as the time required to build an infrastructure after an investment  decision is made.

In this paper, we discuss an analytic recursive method to solve  a refracted optimal multiple stopping problem driven by a  \lev process. This paper focuses on a computational aspect of the optimal multiple stopping problem. It is well known from related studies that  the optimal strategy is of  threshold-type. Consequently,   the optimal stopping  problem reduces to finding  these thresholds.  However, the  determination  of the threshold  values still  involves computing  expectations of a functional  at the end of the constant  refraction period, which is generally not explicit.  In existing literature, Monte Carlo simulation methods are typically   employed  to evaluate these expectations  (see \cite{Bender, Hambly_MF04}, among others). However,  in practice this approach can be  computationally expensive  and even infeasible in terms of the run time.  Furthermore, with  multiple stopping,  one  needs to know the entire expected future payoff functional (with respect to the starting point of the underlying process) in order to  determine backwards these functionals as well as  the optimal threshold levels for earlier stages. The  simulation approach commonly involves  computing these expectations for arbitrarily large number of starting points, and this adds to the computational burden and limits its applicability.  In this regard, it is important to approximate these functions \emph{in closed form} so as to carry out efficiently the backward induction.

One  key  feature of our analysis is that the  rate for discounting future cash flows can be negative or positive.  A negative discount rate can
 accommodate a number of  applications, such as stock loans \cite{Cai2014,xiazhou} as well as  real option problems where the investment  cost grows faster than  the risk-free rate.  In these cases, one can interpret  {that} the effective discount rate is negative.  As argued by  Black \cite{black95} (see also references therein), it is commonly assumed that  the nominal short rate must stay  positive, but  the real interest rate can potentially  be negative, especially  during low-yield  regimes.  Hence, our  framework   permits discounting cash flows at a negative effective or  real interest rate.

In our model, the  underlying process is a spectrally negative \lev process, which has recently been widely used in mathematical finance.
Negative jumps can model sudden downward movements of an asset price. These processes  are suitable in the structural models of credit risk and generate non-zero limiting value of the credit spread as the maturity goes to zero  as studied in \cite{Leung_Yamazaki_2011,Hilberink2002a, Kyprianou2007a, Leung_Yamazaki_2010, Surya_Yamazaki_2012}. Some recent applications of
spectrally negative \lev processes include  the  pricing of perpetual American and exotic options
\cite{Alili2005,Avram_2004},   optimal dividend problems
\cite{Avram_et_al_2007,Kyprianou_Palmowski_2007,Loeffen_2008}, and
capital reinforcement timing \cite{Egami_Yamazaki_2010}.  For related optimal multiple stopping problems under spectrally negative models,  we mention  \cite{ZeghalSwing} for a swing put option with   constant refraction times, and \cite{Yamazaki_2012} with a more general payoff function without refraction times.
For   models with more general processes,
Leung et al.\ \cite{TimKazuHZ14} study  a refracted optimal  multiple stopping problem  driven by a two-sided  \lev process with general random refraction times, and Christensen and Lempa \cite{ChristensenLempa13} consider a similar problem  driven by a general Markov process with exponential refraction times.

Motivated by  the maturity  randomization (Canadization)  method proposed by Carr \cite{Carr_1998}, we provide  an analytical approximation by replacing  every  constant refraction time with an independent Erlang random variable, or a finite sum of independent, identically distributed exponentially distributed times.  Our method involves  repeated  integrations with respect to the \emph{resolvent measure}, which is written in terms of the \emph{scale function} of the underlying spectrally negative \lev process.
For the randomization methods applied in the pricing of finite-time horizon American options, we refer the reader to \cite{Kleinert_Schaik,kypri_Canada}; similar ideas are also used in recent work on the so-called Wiener-Hopf simulation \cite{Kuznetsov_2010_3, Ferreiro-Castilla_2013}. Bouchard et al.  \cite{BouchardELKTouzi2005} analyze a maturity randomization  algorithm and apply it to stochastic control problems with applications to optimal single stopping and 	dynamic hedging under uncertain volatility.

In order to apply the randomization method, the closed form expression must be preserved after the integration is applied with respect to the resolvent measure.  This is satisfied when the Laplace exponent of the underlying \lev process has a rational form, in which case the scale function can be written as a {finite} sum of (possibly complex) exponentials (see \cite{Kuznetsov_2011}).  Here, we focus on \emph{phase-type} \lev processes \cite{Asmussen2004}, which constitute an important class of \lev processes   with Laplace exponents of rational transform.  In principle, any spectrally negative \lev process can be approximated by a \lev process of this form (which we call \emph{phase-type fitting}). In particular, Egami \& Yamazaki \cite{Egami_Yamazaki_2010_2} give a series of numerical experiments for approximating the scale function of a general spectrally negative \lev process by that of a phase-type \lev process.  For a hyperexponential fitting method applied to a CGMY process with a completely monotone \lev density, see \cite{Asmussen2007a}.

Motivated by these, we combine phase-type fitting and randomization methods to compute efficiently the solutions of the optimal multiple stopping problem.   Specifically, given a general spectrally negative \lev process, we first approximate it by a phase-type \lev process, and then approximate the solutions by randomizing the constant refraction times using independent, identically distributed Erlang random variables.   We shall show that the resulting approximating value functions are  written in  closed form, with  the associated parameters  computed recursively.

 {Our objective is to evaluate numerically the effectiveness of  our approach, especially the accuracy of the value functions as a result of  (i) phase-type   fitting of the jump distribution,  and (ii)  refraction times randomization. Regarding  part (i), while it is theoretically known that the class of phase-type distributions is dense in the class of all positive-valued distributions,  there does not currently exist a single algorithm that can produce a sequence of phase-type distributions that are guaranteed to converge to a desired distribution (unless it has a completely monotone density).  As for   refraction times randomization, we refer to \cite{BouchardELKTouzi2005} and \cite{levendorski2011convergence} for the    related convergence results on the randomization approach. In a related study \cite{Kleinert_Schaik}, detailed numerical experiments are conducted to confirm the convergence for pricing American put options when the \lev process is in the meromorphic class. It is noted, however, that these results do not apply directly to our case, because we deal with a multiple optimal stopping problem where the refraction time is randomized.  In addition, the payoff function (call type) is not bounded. For these reasons, it is important that our approach is evaluated numerically.}

%We evaluate the effectiveness of  our approach, especially the accuracy of the value functions as a result of  (i) phase-type   fitting of the jump distribution,  and (ii)  refraction times randomization.  
In this paper, through a series of numerical examples, we show that our method is capable of accurately and  efficiently computing  the sequence of value functions and optimal exercise thresholds.   In addition, the run-time analysis shows that this approach is significantly faster than the  {Monte Carlo} simulation methods  {that adopt the Euler's method to approximate the expected value function of the next stage at the constant refraction time}.  On the other hand,  as the number of stages and the shape parameter of the Erlang distribution increase, the usual machine double precision may not be  capable of computing the parameters in the value functions.  Barring this potential issue,  our closed form formulas are confirmed by  {comparing with the simulated values} and allow for more efficient computation.

The rest of the paper is organized as follows. Section  \ref{section_preliminaries} reviews the optimal multiple stopping problem for a spectrally negative \lev process and the characterization of the optimal strategies in terms of up-crossing times.   Section \ref{section_recursive} presents  our  randomization method and  {derives}  the analytic value functions  recursively using the resolvent measure.  Section \ref{section_phase_type} discusses the  applications of  phase-type fitting and   {shows} a backward induction formula to compute the parameters of the  value functions for the randomized problem.   We conclude in Section \ref{section_numerical_results} with numerical evaluation of our proposed method.

 \section{Preliminaries} \label{section_preliminaries}
Let  $(\Omega, \F, \p)$  be a probability space hosting a spectrally negative \lev process $X=(X_t)_{t\ge 0}$.  We define $\mathbb{F}:=(\F_t)_{t\ge0}$ as the completed filtration generated by $X$, and $\setT$  the set of all $[0,\infty]$-valued $\mathbb{F}$-stopping times.  We denote  $\p_x$ as the  probability  and   $\E_x$ as the expectation with initial value $X_0 = x \in \R$.  In particular, when $X_0=0$, we drop the subscripts in $\p_x$ and $\E_x$.

 {By the L\'evy-Khintchine formula, $X$ can be} characterized by its \emph{Laplace exponent}  given by
\begin{align}
\psi(s)  := \log \E \left[ e^{s X_1} \right] =  c s +\frac{1}{2}\sigma^2 s^2 + \int_{(-\infty,0)} (e^{s z}-1 - s z 1_{\{-1 < z < 0\}}) \Pi (\diff z), \quad s \geq 0, \label{laplace_spectrally_positive}
\end{align}
where $\Pi$ is a \lev measure with the support $(-\infty, 0)$ that satisfies the integrability condition $\int_{(-\infty,0)} (1 \wedge z^2) \Pi(\diff z) < \infty$.  It has paths of bounded variation if and only if $\sigma = 0$ and $\int_{( -1,0)}|z|\, \Pi(\diff z) < \infty$; in this case, we write \eqref{laplace_spectrally_positive} as
\begin{align}
\psi(s)   =  \widetilde{c} s +\int_{(-\infty,0)} (e^{s z}-1 ) \Pi (\diff z), \quad s \geq 0,
\end{align}
with $\widetilde{c} := c - \int_{( -1,0)}z\, \Pi(\diff z)$.  We exclude the case in which $-X$ is a subordinator (i.e., $X$ has {monotonically decreasing} paths a.s.). This assumption implies that $\widetilde{c}> 0$ when $X$ is of bounded variation.  {In addition,  we assume throughout the paper that $X_t$ admits a density; this is guaranteed to be satisfied if $\sigma > 0$ or the absolutely continuous part of the \lev measure has an infinite mass (see e.g.\ \cite{Tucker}).
\begin{assump} \label{assump_absolutely_continuous}
We assume that $\p \{X_t  \in \diff x \} \ll \diff x$ for all $t > 0$.
\end{assump}}

We consider the problem with sequential stopping (exercise) opportunities, where the payoff from each  exercise is
\begin{align}\phi(x):=\mathrm{e}^x-K, \quad  x \in \R,
\end{align}
for some constant $K > 0$.  The associated  optimal multiple stopping problem is defined as
\be
v^{(N)}(x) :=\sup_{\vec{\tau}\in \setT^{(N)}}\ex_x\left[\sum_{n=1}^N \mathrm{e}^{-\alpha\tau_n}\phi(X_{\tau_n})\ind_{\{\tau_n<\infty\}}\right], \quad  x \in \R, \; N \in \N. \label{Problem}
\ee
Here, the optimization is over all increasing sequence of stopping times in such a way that any two consecutive stopping times are separated by a  refraction period $\delta > 0$.  In other words, the set of admissible strategies is given by
\begin{align}
\setT^{(N)} :=\{\vec{\tau}=(\tau_N,\ldots,\tau_1) \in \setT^N \,:\,    \tau_{n+1}+\delta \le \tau_n, n = N-1, \ldots, 1\}.
\end{align}
Here we label in such a way that $\tau_n$ is the stopping time when there are $n$ stopping opportunities  left.

% \be
%v^{(n)}(x):=\sup_{ {\tau}\in \setT}\ex_x\!\left[\mathrm{e}^{-\alpha\tau}\phi^{(n)}(X_{\tau})\ind_{\{\tau<\infty\}}\right]\label{single1}
%\ee
%where $\alpha\in\R$ is the discount rate, $\setT$ is the set of $\mathbb{F}$-stopping times, and  \be \phi^{(n)}(x) :=\phi(x)+\ex_x\!\left[\,\mathrm{e}^{-\alpha\delta}v^{(n-1)}(X_\delta)\,\right],\quad n = 1, 2, \ldots, N,\label{single2}\ee and  $v^{(0)}(x) := 0.$ As such, the optimal multiple stopping problem involves finding an increasing sequence of stopping times such that  any two consecutive stopping times are separated by a  refraction period $\delta > 0$.

%
%\be
%v^{(N)}(x) :=\sup_{\vec{\tau}\in \setT^{(N)}}\ex_x\left[\sum_{n=1}^N \mathrm{e}^{-\alpha\tau_n}\phi(X_{\tau_n})\ind_{\{\tau_n<\infty\}}\right], \quad  x \in \R, \; N \in \N. \label{Problem}
%\ee
%Here, the optimization is over all increasing sequence of stopping times in such a way that any two consecutive stopping times must be separated by a constant refraction time $\delta > 0$.  In other words, the set of admissible strategies is given by
%\begin{align*}
%\setT^{(N)} :=\{\vec{\tau}=(\tau_N,\ldots,\tau_1) \in \setT^N \,:\,    \tau_{n+1}+\delta \le \tau_n, n = N-1, \ldots, 1\}.
%\end{align*}
%Here we label in such a way that $\tau_n$ is the stopping time when there are $n$ stopping opportunities  left. \tim{We ought to be careful about the stopping times. I know this is correct as is for constant refraction times, but our solution method involves random $\delta$. }

Next, for any given discount rate $\alpha \in \R$, we define the process
\begin{align}
X_t^{(\alpha)} := X_t - \alpha t, \quad t \geq 0,
\end{align}
 which is  {either} a spectrally negative \lev process  {or the negative of a subordinator}.
As is well  known, the limit of the running supremum  $\overline{X}_\infty^{(\alpha)} := \sup_{0 \leq t < \infty} X_t^{(\alpha)}$ is a $\p$-exponential random variable with rate  parameter
\begin{align}
\widetilde{\Phi}{(\alpha)} := \sup \left\{ \lambda \geq 0: \psi(\lambda) - \alpha \lambda   { \le } 0\right\}, \end{align}
with the convention that $\overline{X}_\infty^{(\alpha)} =\infty$ a.s.\  when $\widetilde{\Phi}{(\alpha)} = 0$,  {and that $\overline{X}_\infty^{(\alpha)} =0$ a.s.\ when $\widetilde{\Phi}{(\alpha)} = \infty$}. Hence, if $\widetilde{\Phi}(\alpha) > 0$,  then the  expectation
\begin{align}
\E [e^{\varrho \overline{X}^{(\alpha)}_\infty}] = \Big( 1 - \frac \varrho {\widetilde{\Phi}(\alpha)}\Big)^{-1} < \infty,
\end{align}
for any  constant $\varrho \in (0,\widetilde{\Phi}(\alpha))$.
Since  $\lambda \mapsto \psi(\lambda) - \alpha \lambda$ is strictly convex on $[0,\infty)$ and is zero at the origin, $\widetilde{\Phi}(\alpha) > 1$  {if} $\psi(1) < \alpha$.  In this case, we can choose $\hat{\varrho} > 1$ such that the above moment generating function is finite, and thus, with the positivity of $K$, we have 
\be\label{eq:ass}
\ex_x\!\left[\,\left(\Big( \sup_{0\le t<\infty} \mathrm{e}^{-\alpha t}\phi(X_t) \Big)^+ \right)^{\hat{\varrho}}\,\right]<\infty, \quad x\in\R.
\ee
This  is a critical condition so that the solution of the problem is nontrivial (see, e.g., \cite{Carmona2008, TimKazuHZ14, ZeghalSwing}).  In fact, we can slightly weaken the condition for the case $\alpha < 0$ (where $\exp (-\alpha t) K$ grows to infinity) to accommodate the case $\psi(1) = \alpha$ given $\psi'(1) < 0$ (see  \cite{TimKazuHZ14} for a proof).

\begin{assump}  \label{assump_alpha} We assume that either (i) $\psi(1) < \alpha$ or (ii) $\psi(1) = \alpha < 0$ and $\psi'(1) < 0$ holds.
\end{assump}

This guarantees that the value function is finite and  admits a nontrivial solution. The optimal strategy is given by a sequence of \emph{up-crossing} times of the form
\begin{align} \label{optimal_solutions_form}
\begin{split}
\tau^*_N &:= T_{a_N^*}^+, \\
\tau^*_{n} &:= T_{a_n^*}^+ \circ \theta_{\tau^*_{n+1} + \delta} + \tau^*_{n+1} + \delta, \quad 1 \leq n \leq N-1,
\end{split}
\end{align}
for some parameters $a^* = (a_n^*)_{1 \leq n \leq N}$ where $\theta$ is the time-shift operator and
\begin{align}
T_a^+ := \inf \left\{ t > 0:  X_t \geq a \right\}, \quad a \in \R,
\end{align}
with the usual convention that $\inf \varnothing = \infty$.
 The optimality of the threshold strategy for the single stopping problem has been shown by Mordecki \cite{mordecki2002}. The same  characterization  holds for the multi-stage problem, and we  refer the reader to \cite{Carmona2008,ZeghalSwing} and the authors' companion paper \cite{TimKazuHZ14} for the proof.

In view of these characterizations, the implementation of the optimal strategy reduces to identifying the values of $a^*$, and this is the primary objective of our paper.  To  this end, we first rewrite  \eqref{Problem} recursively as follows (see \cite{kobylanski2011,TimKazuHZ14}, among others):
\be
v^{(n)}(x):=\sup_{ {\tau}\in \setT}\ex_x\!\left[\mathrm{e}^{-\alpha\tau}\phi^{(n)}(X_{\tau})\ind_{\{\tau<\infty\}}\right],\label{single1}
\ee
where \be \phi^{(n)}(x) :=\phi(x)+\ex_x\!\left[\,\mathrm{e}^{-\alpha\delta}v^{(n-1)}(X_\delta)\,\right],\quad n = 1, 2, \ldots, N,\label{single2}\ee and  $v^{(0)}(x) := 0.$

Given that  an optimal stopping time in \eqref{single1} is of  threshold type, the value of $a_n^*$ can be determined  by maximizing the value function over candidate threshold values:
\begin{align}
a_n^* \in \arg \max_{a \in \R}v^{(n)}_a(x), \label{def_a_max}
\end{align}
(which maximizes uniformly in $x \in \R$), where
\begin{align}
v^{(n)}_a(x) := \ex_x\!\left[\mathrm{e}^{-\alpha T_a^+}\phi^{(n)}(X_{T_a^+})\ind_{\{T_a^+<\infty\}}\right], \quad a,x \in \R.
\end{align}

The following lemma is well-known for positive discount rate  $\alpha \geq 0$ (see \cite{Kyprianou2006}, Theorem 3.12). Under Assumption \ref{assump_alpha}, we generalize the result to accommodate  the case with $\alpha < 0$.  Let
\begin{align}
\Phi(\alpha) := \sup \left\{ \lambda \geq 1: \psi(\lambda) = \alpha\right\},
\end{align}
which is guaranteed to exist by Assumption \ref{assump_alpha} (which postulates that $\psi(1) \leq \alpha$) and because $\psi$ is strictly convex on $[1, \infty)$.
\begin{lemma} \label{lemma_laplace_hitting_time}  {Under Assumption \ref{assump_alpha},} we have $\ex  [\mathrm{e}^{-\alpha T_y^+} \ind_{\{T_y^+<\infty\}}]  =  \exp (-\Phi(\alpha)y)$ for $y > 0$ and equals $1$ otherwise.
\end{lemma}
\begin{proof}
We shall show for  $\psi(1) < \alpha$ ({and hence} $\Phi(\alpha) > 1$  by the strict convexity of $\psi$); the case $\psi(1) = \alpha$ then follows immediately by the monotone convergence theorem and the continuity of $\Phi(\cdot)$.

Fix $y > 0$. The process $(\exp (\Phi(\alpha) X_t - \alpha t))_{t \geq 0}$ is a martingale (see {(3.11)} of \cite{Kyprianou2006}), and hence we can derive (as in the first part of the proof of \cite{Kyprianou2006}, Theorem 3.12) that
\begin{align}
\E \left[ e^{\Phi(\alpha) X_{t \wedge T^+_y} - \alpha (t \wedge T^+_y)}\right] = 1, \quad t \geq 0. \label{martingale_esccher}
\end{align}
Here the integrand of the left-hand side is bounded in $t$ by an integrable random variable, i.e.,
\begin{align}
e^{\Phi(\alpha) X_{t \wedge T^+_y} - \alpha (t \wedge T^+_y)} =e^{(\Phi(\alpha)-1) X_{t \wedge T^+_y}}e^{X_{t \wedge T^+_y}^{(\alpha)}} \leq  e^{(\Phi(\alpha)-1) X_{t \wedge T^+_y}}e^{\overline{X}_\infty^{(\alpha)}}\leq e^{(\Phi(\alpha)-1) y}e^{\overline{X}_\infty^{(\alpha)}},
\end{align}
where the last inequality holds because $\Phi(\alpha) > 1$ and $X_{t \wedge T_y^+} \leq y$ a.s.\ due to the lack of positive jumps.
Hence applying dominated convergence in \eqref{martingale_esccher} gives
\begin{align}
1=\E \left[ e^{\Phi(\alpha) X_{T^+_y} - \alpha T^+_y} 1_{\{ T_y^+ < \infty\}}\right] =e^{\Phi(\alpha) y} \E \left[ e^{- \alpha T^+_y} 1_{\{ T_y^+ < \infty\}}\right],
\end{align}
where the last equality holds as $X_{T_y^+} = y$ on $\{ T_y^+ < \infty \}$.  This completes the proof.
\end{proof}

Due to Lemma \ref{lemma_laplace_hitting_time} and to the fact that the process $X$ necessarily creeps upward and hence $X_{T_a^+} = a$ on $\{ T_a^+ < \infty \}$ under $\p_x$ with $x \leq a$, we can write
\begin{align} \label{v_k_expression}
v^{(n)}_a(x) = \left\{ \begin{array}{ll}e^{-\Phi(\alpha)(a-x)} \phi^{(n)}(a), & x < a, \\ \phi^{(n)}(x), & x \geq a. \end{array}\right.
\end{align}

\begin{remark} \label{remark_monotonicity}
It can be shown that the threshold levels are bounded from below by $\log K$ and increase as the number of remaining stopping opportunities decreases, i.e., $\log K < a_N^* \leq \cdots \leq a_{1}^*$. It has been shown in \cite{TimKazuHZ14} that this monotonicity also holds when the refraction times $\delta$'s are generalized to be independent, identically distributed random variables provided that  they are independent of $X$, and $X_\delta$ admits a density.  They also show that there  exists a limit $a^*_\infty := \lim_{N \rightarrow \infty}a^*_N \geq \log K$. \end{remark}

\section{Recursive Analytic Formula} \label{section_recursive}
The characterization of the optimal strategy as described in the previous section greatly simplifies the problem.  In practice, however, the solution cannot be obtained analytically because in general the distribution of $X_\delta$ is not known in view of \eqref{single1}.
The biggest hurdle therefore is to compute the expectation
\begin{align}
\E_x \left[ e^{- \alpha \delta}v^{(n-1)} (X_\delta) \right], \quad 2 \leq n \leq N. \label{expectation_at_refraction}
\end{align}
 In order to circumvent  this difficulty, we adopt the  Canadization technique by Carr \cite{Carr_1998} and approximate \eqref{expectation_at_refraction} by replacing the constant $\delta$ with some independent Erlang random variable $\eta(M,\lambda)$, or equivalently a sum of $M$  independent, identically distributed exponential random variables with parameter $\lambda$.   Herein, we set
\begin{align}
\lambda = \lambda^{(M)} := M/\delta.
\end{align}
  Then, $\eta(M, \lambda) \approx \delta$ for large $M$ by the strong law of large numbers. In other words, we solve the randomized version of the optimal multiple stopping problem in order to approximate the one with constant refraction times. For optimal stopping problems with  random refraction times, the  filtration needs to be modified; see \cite{ChristensenLempa13} for the precise construction. However, this technical detail  does not affect the resulting threshold structure of the optimal stopping strategies, as discussed in  \cite{ChristensenLempa13,TimKazuHZ14}.

In this section, we shall show that the  value functions with randomized refraction times   can be obtained recursively via the resolvent measure written in terms of the scale function.

\subsection{First Step} We first construct the base case.
 In view of  \eqref{v_k_expression} for $n=1$, because $\phi^{(1)} \equiv \phi$, the value of $a_1^*$ is obtained analytically via \eqref{def_a_max}. The first order condition becomes
\begin{align}
0 = \phi'(a)  - \Phi(\alpha) \phi(a),
\end{align}
which admits a unique solution given by
\begin{align}
a^*_1 = \log \frac {\Phi(\alpha) K} {\Phi(\alpha)-1}. \label{a_1_star}
\end{align}
It is easy to check that this is equivalent to the smooth fit condition $v^{(1)'}(a^*_1+)=v^{(1)'}(a^*_1-)$.

We now start at
\begin{align} \label{u_1_0}
u^{(1,0)}(x) := v^{(1)}(x) = \left\{ \begin{array}{ll} \phi(x) , & x \geq a^*_1, \\ \phi(a^*_1) e^{-\Phi(\alpha)(a^*_1-x)}, & x < a^*_1,
\end{array} \right.
\end{align}
and derive an analytical expression for the expectation
\begin{align}
\E_x \left[ e^{- \alpha \eta(M, \lambda)}v^{(1)} (X_{\eta(M, \lambda)}) \right]
\end{align}
as an approximation of \eqref{expectation_at_refraction} for $n=2$.

The very initial task is to compute the case of exponential time horizon,
\begin{align}
u^{(1,1)}(x) := \E_x \left[ e^{- \alpha \eta(1,\lambda)}v^{(1)} (X_{\eta(1,\lambda)}) \right]  = \lambda \int_0^\infty  e^{-(\lambda+\alpha) t} \E_x \left[ v^{(1)} (X_{t}) \right] \diff t
= \mathcal{M}_x  u^{(1,0)}, \label{first_recursion}
\end{align}
where we define, for any measurable $f$, {that (whenever it exists)}
\begin{align}
\mathcal{M}_x  f := \lambda \int_\R   \Theta^{(\lambda+\alpha)}(x, \diff y) f(y) \label{operator_M}
\end{align}
for a resolvent measure
\begin{align}
\Theta^{(q)}(x, \diff y) := \int_0^\infty e^{-q t} \p_x \left\{ X_t \in \diff y \right\} \diff t, \quad y \in \R \; \textrm{and} \; q > 0.
\end{align}

It is known for the case of spectrally negative \lev process that this resolvent measure admits a density and can be written in terms of the so-called \emph{scale function}.  Fix $q \ge 0$,  the  ($q$-)scale function,
\begin{align}
W^{(q)}: \R \rightarrow [0,\infty),
\end{align}
 is zero on $(-\infty,0)$, continuous and strictly increasing on $[0,\infty)$, and is characterized by the Laplace transform:
\begin{align}
\int_0^\infty e^{-s x} W^{(q)}(x) \diff x = \frac 1
{\psi(s)-q}, \qquad s > \lapinv,
\end{align}
where
\begin{equation}
\lapinv :=\sup\{\lambda \geq 0: \psi(\lambda)=q\}. 
\end{equation}
By Corollary 8.9 of \cite{Kyprianou2006}, the resolvent measure $\Theta^{(q)}(x,\diff y)$ has a density $\theta^{(q)}(y-x)$ with respect to the Lebesgue measure where
\begin{align}
\theta^{(q)}(z)  := \Phi'(q) e^{- \Phi(q) z} - W^{(q)}(-z), \quad z \in \R. \label{resolvent_density_general}
\end{align}
Hence on condition that $\lambda + \alpha > 0$, \eqref{first_recursion} can be rewritten using the scale function.

This manipulation can be applied repeatedly by further adding more  independent, identically distributed exponential random time horizons.
Indeed, for any $2 \leq m \leq M$ and $x \in \R$,
\begin{multline}
u^{(1,m)}(x) := \E_x \left[ e^{- \alpha \eta(m,\lambda)}v^{(1)} (X_{\eta(m,\lambda)}) \right]  = \lambda \int_0^\infty  e^{-(\lambda+\alpha) t}\E_x [u^{(1,m-1)}(X_t) ] \diff t  \\
=  \mathcal{M}_x u^{(1,m-1)} =  \cdots = \mathcal{M}_x^{m} u^{(1,0)}.
\end{multline}

\subsection{Multiple Steps}

Now that \eqref{expectation_at_refraction} is approximated by $u^{(1,M)}(x)$ we can approximate $\phi^{(2)}(x)$ as in \eqref{single2} by
\begin{align}
\widetilde{\phi}^{(2)}(x) &:= \phi(x) + u^{(1,M)}(x).
\end{align}
Using this approximation, we can obtain an approximation to $a_2^*$, say $\widetilde{a}_2^*$, by the first order condition
$0= \widetilde{\phi}^{(2)'}(a) - \Phi(\alpha) \widetilde{\phi}^{(2)}(a)$,
and obtain an approximation to $v^{(2)}$:
\begin{align}
u^{(2,0)}(x) &\equiv \widetilde{v}^{(2)}(x) := \left\{ \begin{array}{ll} \widetilde{\phi}^{(2)}(x) , & x \geq \widetilde{a}^*_2, \\ \widetilde{\phi}^{(2)}(\widetilde{a}^*_2) e^{-\Phi(\alpha)(\widetilde{a}^*_2-x)}, & x < \widetilde{a}^*_2.
\end{array} \right.
\end{align}
Similarly to the first step,  for any $1 \leq m \leq M$,
\begin{align}
u^{(2,m)}(x) := \E_x \left[ e^{- \alpha \eta(m,\lambda)}\widetilde{v}^{(2)} (X_{\eta(m,\lambda)}) \right] =  \mathcal{M}_x u^{(2,m-1)} = \cdots = \mathcal{M}_x^{m} u^{(2,0)},
\end{align}
which gives an approximation for \eqref{expectation_at_refraction} for $n=3$.

Continuing in this fashion, we can derive the approximations defined by
\begin{align}
\widetilde{\phi}^{(n)}(x) &:= \phi(x) + u^{(n-1,M)}(x), \label{recursion_u_n_0_part1} \\
\widetilde{a}_n^* &\in \arg \{a \in \R : \widetilde{\phi}^{(n)'}(a) - \Phi(\alpha) \widetilde{\phi}^{(n)}(a) = 0 \}, \label{recursion_u_n_0_part2} \\
u^{(n,0)} (x) &\equiv \widetilde{v}^{(n)}(x) := \left\{ \begin{array}{ll} \widetilde{\phi}^{(n)}(x) , & x \geq \widetilde{a}^*_n, \label{recursion_u_n_0_part3} \\ \widetilde{\phi}^{(n)}(\widetilde{a}^*_n) e^{-\Phi(\alpha)(\widetilde{a}^*_n-x)}, & x < \widetilde{a}^*_n,
\end{array} \right. \\
u^{(n,m)}(x) &:= \E_x \left[ e^{- \alpha \eta(m,\lambda)}\widetilde{v}^{(n)} (X_{\eta(m,\lambda)}) \right] =   \mathcal{M}_x^{m} u^{(n,0)}, \quad 1 \leq m \leq M.
\label{recursion_u_n_0_part4}
\end{align}
Finally, $\widetilde{v}^{(N)}(x)$ is the desired approximation to our multiple stopping problem.
For the rest of the paper, we let $\widetilde{a}^*_1 := a^*_1$ for notational convenience.

\section{Spectrally Negative Phase-type Case}\label{section_phase_type}
In order to carry out the algorithm described in the previous section, it is important that the backward induction can be done analytically.  That is to say, the closed-form expression must be preserved after the operator $\M$ as in \eqref{operator_M} is applied.  In this section, we shall show that this is possible if we focus on the phase-type \lev process of the form \eqref{form_levy_compound_poisson} below.

It is known from  {Proposition 1} of \cite{Asmussen2004}   that, for any spectrally negative \lev process $X$, there exists a sequence of spectrally negative phase-type \lev processes $X^{(n)}$ converging to $X$ in the space  $D[0,\infty)$ of real-valued  right-continuous functions   with left-limits (c\`{a}dl\`{a}g); this implies that  $X_1^{(n)} \rightarrow X_1$ in distribution (see Remark 1 of \cite{Asmussen2004}  and Corollary VII 3.6 of \cite{Jacod_Shirayev_2003}).  From this, we naturally conjecture that the resolvent of  $X$ can be approximated by that of a phase-type \lev process.  Indeed, Egami \& Yamazaki \cite{Egami_Yamazaki_2010_2} show numerically that the scale function of a general spectrally negative \lev process can be approximated at least when the \lev measure is finite.   For the case of infinite \lev measure satisfying the condition of Asmussen and Rosi\'nski \cite{Asmussen_Rosinski_2001}, Asmussen et al.\ \cite{Asmussen2007a} show that a Brownian motion can be used as a proxy to  approximate the frequent infinitesimal jumps (where they consider a \lev measure with a completely monotone density).  In the next section, we analyze the approximation errors through numerical experiments. 

Throughout this section, let $X$ be a spectrally negative phase-type \lev process,
\begin{equation}
  X_t  - X_0= \widetilde{c} t+\sigma B_t - \sum_{n=1}^{N_t} Z_n, \quad 0\le t <\infty, \label{form_levy_compound_poisson}
\end{equation}
for some $\widetilde{c} \in \R$ and $\sigma \geq 0$.  Here $B=(B_t)_{ t\ge 0}$ is a standard Brownian motion, $N=(N_t)_{t\ge 0}$ is a Poisson process with arrival rate $\rho$, and  $Z = (Z_n)_{n = 1,2,\ldots}$ is an independent, identically distributed sequence of phase-type-distributed random variables with representation $(d,{\bm \alpha},{\bm T})$.  These processes are assumed to be mutually independent.
Recall that a  distribution on $(0,\infty)$ is of phase-type  if it is the distribution of the absorption time  in a finite state continuous-time Markov chain consisting of one absorbing state and $d\in\mathbb{N}$ transient states. Thus,  any phase-type distribution can be represented by $d$, the $d\times d$ transition intensity matrix over all  transient states $\bm{T}$, and the initial distribution of the Markov chain $\bm{\alpha}$.

{Let $\mathbf{t}$ be the transition probabilities from the $d$ transient states to the absorbing state.} The Laplace exponent \eqref{laplace_spectrally_positive} is then
\begin{align}
 \psi(s)   = \widetilde{c} s + \frac 1 2 \sigma^2 s^2 + \rho \left( {\bm \alpha} (s {\bm I} - {\bm{T}})^{-1} {\mathbf{t}} -1 \right),
 \end{align}
which can be extended to $s \in \mathbb{C}$ except at the negative of eigenvalues of ${\bm T}$.

For the rest, let us define
\begin{align}
p := \alpha + \lambda = \alpha + M/\delta,
\end{align}
and we assume this to be strictly positive.

 Suppose $\{ -\xi_{i,p}; i \in \mathcal{I}_p \}$ is the set of the roots of the equality $\psi(s) = p$ with negative real parts. {As is discussed in Section 5.4 of \cite{Kuznetsov_2011} and has been confirmed numerically in \cite{Egami_Yamazaki_2010_2} {(see also our numerical results in Section  \ref{section_numerical_results}
 of this current paper)}, it is highly unlikely that any root in $\mathcal{I}_p$ has multiplicity larger than one. Hence, we can assume that the roots in $\mathcal{I}_p$ are distinct. Consequently,} 
the scale function can be written
\begin{align} \label{scale_function_simple}
\begin{split}
W^{(p)}(x) &= \left\{ \begin{array}{ll}\Phi'(p)  e^{\Phi(p) x} - \sum_{i \in \mathcal{I}_p} \kappa_{i,p}  e^{-\xi_{i,p}x}, & x \geq 0, \\
0, & x < 0,
\end{array} \right.
\end{split}
\end{align}
where
\begin{align}
\kappa_{i,p} &:= \left. \frac { s+\xi_{i,p}} {p-\psi(s)} \right|_{s = -\xi_{i,p}} = - \frac 1 {\psi'(-\xi_{i,p})};
\end{align}
see \cite{Egami_Yamazaki_2010_2}.  Here $\{ \xi_{i,p}; i \in \mathcal{I}_p \}$ and  $\{ \kappa_{i,p}; i \in \mathcal{I}_p \}$ are possibly complex-valued.
Hence the resolvent density \eqref{resolvent_density_general} is written
\begin{align}
\theta^{(p)}(z)  = \left\{ \begin{array}{ll}   \Phi'(p) e^{- \Phi(p) z}, & z > 0, \\ \sum_{i \in \mathcal{I}_p} \kappa_{i,p}  e^{\xi_{i,p}z}, &  z \leq 0. \end{array} \right.
\end{align}

%\begin{remark} \label{remark_multiplicity_unlikely} \red{[HZ: to be removed as this has been moved earlier.]}As is discussed in Section 5.4 of \cite{Kuznetsov_2011} and has been confirmed numerically in \cite{Egami_Yamazaki_2010_2}, it is highly unlikely that any root in $\mathcal{I}_p$ has multiplicity larger than one.  Hence, our assumption that the roots in $\mathcal{I}_p$ are distinct  is indeed minor.
%\end{remark}

Our objective here is to show that  the function $u^{(n,m)}$ for each $n \geq 1$ and $0 \leq m \leq M$ (derived recursively as in the previous section) is a piecewise function with subdomains {$(\widetilde{a}^*_{l},\widetilde{a}^*_{l-1})_{1 \leq l \leq n+1}$} (see Remark \ref{remark_monotonicity} regarding the monotonicity of $\widetilde{a}^*_n$) where we define, for notational convenience, $\widetilde{a}_{0}^* := \infty$ and $\widetilde{a}_{n+1}^* := -\infty$. More specifically, we shall show that,  on each subdomain, it is a sum of products of polynomials and exponentials:
\begin{align} \label{u_form}
u^{(n,m)}(x) = f^{(n,m,l)}(x)
\end{align}
for {$\widetilde{a}^*_{l} < x < \widetilde{a}^*_{l-1}$} where we define
\begin{align}
\begin{split}
f^{(n,m,l)}(y) &:= A^{(n,m,l)} + B^{(n,m,l)} e^y +  \sum_{i \in \mathcal{I}_p} \sum_{h=0}^{I_{n,m}} (C^{(n,m,l)}_{i,h} e^{- \xi_{i,p} y} y^h) \\ & \qquad +  \sum_{h=0}^{I_{n,m}} (D^{(n,m,l)}_h e^{\Phi(p) y} y^h) + E^{(n,m,l)} e^{\Phi(\alpha) y}, \quad 1 \leq l \leq n+1, \; y \in \R,
\end{split}
\end{align}
with $I_{n,m} := (n-1)M + m-1$.
Here the parameter set
\begin{align}
\Gamma_{n,m} := (A^{(n,m,l)}, B^{(n,m,l)}, \{ C^{(n,m,l)}_{i,h}, i \in \mathcal{I}_p, 0 \leq h \leq I_{n,m}  \}, \{ D^{(n,m,l)}_h, 0 \leq h \leq I_{n,m} \}, E^{(n,m,l)})_{1 \leq l \leq n+1}, \label{parameter_set}
\end{align}
satisfies
\begin{align} \label{hypothesis_corner}
\begin{split}
&D^{(n,m,1)} = E^{(n,m,1)} = 0,  \\
&A^{(n,m,n+1)} = B^{(n,m,n+1)}=C^{(n,m,n+1)}  = 0.
\end{split}
\end{align}
%and can be computed recursively by the algorithm we present below.

  Its proof and the derivation of the parameter set $\Gamma_{n,m}$ can be done inductively.
Along the same line as the arguments in the last section, we go through the backward induction. First, the base case ($n=1$ and $m=0$) is trivial because, in view of \eqref{u_1_0},
the function $u^{(1,0)}$ can be written as \eqref{u_form} by setting $A^{(1,0,1)} = -K$, $B^{(1,0,1)} = 1$, $E^{(1,0,2)} = \phi(\widetilde{a}^*_1) \exp (-\Phi(\alpha)\widetilde{a}^*_1)$ with $I_{n,m} = -1$.

In view of our discussion in the previous section, there are two types of inductive steps.  The first kind increments the step counter $n$ while the second kind increments $m$ by applying the integration with respect to the resolvent measure.  We shall call the former  Step I and the latter Step II.

\subsection{Inductive Step I}
We show that if the hypothesis holds for $n \geq 1$ and $m =M$, then it also holds for some $n+1$ and $m=0$.   By this hypothesis, the equations \eqref{recursion_u_n_0_part1} and \eqref{u_form} give, for {$\widetilde{a}^*_{l} < x < \widetilde{a}^*_{l-1}$},
\begin{multline}
\widetilde{\phi}^{(n+1)}(x) = (A^{(n,M,l)}-K) + (B^{(n,M,l)}+1) e^x +  \\ \sum_{i \in \mathcal{I}_p} \sum_{h=0}^{I_{n,M}} (C^{(n,M,l)}_{i,h} e^{- \xi_{i,p} x} x^h) +  \sum_{h=0}^{I_{n,M}} (D^{(n,M,l)}_h e^{\Phi(p) x} x^h) + E^{(n,M,l)} e^{\Phi(\alpha) x},
\end{multline}
and
\begin{multline}
\widetilde{\phi}^{(n+1)'}(x) = (B^{(n,M,l)}+1) e^x +  \sum_{i \in \mathcal{I}_p} \sum_{h=0}^{I_{n,M}} C^{(n,M,l)}_{i,h} (- \xi_{i,p} e^{- \xi_{i,p} x} x^h + h e^{- \xi_{i,p} x} x^{h-1}) \\ +  \sum_{h=0}^{I_{n,M}} D^{(n,M,l)}_h (\Phi(p) e^{\Phi(p) x} x^h + h e^{\Phi(p) x} x^{h-1}) + \Phi(\alpha) E^{(n,M,l)} e^{\Phi(\alpha) x}.
\end{multline}
By \eqref{recursion_u_n_0_part2}, we can identify the optimal threshold $\widetilde{a}^*_{n+1}$.

Now, in view of \eqref{recursion_u_n_0_part3}, the representation of $u^{(n+1,0)}$ can be obtained by setting
\begin{align}
E^{(n+1,0,n+2)} = \widetilde{\phi}^{(n+1)}(\widetilde{a}^*_{n+1}) e^{-\Phi(\alpha)\widetilde{a}^*_{n+1}},
\end{align} and $A^{(n+1,0,n+2)} = B^{(n+1,0,n+2)} =C^{(n+1,0,n+2)} =D^{(n+1,0,n+2)} =0$,
and for $1 \leq l \leq n+1$
\begin{align}
A^{(n+1,0,l)} = A^{(n,M,l)} -K, \; B^{(n+1,0,l)} = B^{(n,M,l)} + 1
\end{align}
and
\begin{align}
C^{(n+1,0,l)} = C^{(n,M,l)}, \; D^{(n+1,0,l)} = D^{(n,M,l)},\; E^{(n+1,0,l)} = E^{(n,M,l)}.
\end{align}
It can be confirmed that because, by assumption, $D^{(n,M,1)} = E^{(n,M,1)} = 0$,
we have $D^{(n+1,0,1)} = E^{(n+1,0,1)} = 0$.
\subsection{Inductive Step II}  It is now sufficient to show that, for fixed $n \geq 1$ and $0 \leq m \leq M-1$, the hypothesis for $m$ implies that for $m +1$.

For every $x \in \R$, with $1 \leq L (= L_x) \leq {n+1}$ be such that $\widetilde{a}_{L}^* < x < \widetilde{a}_{L-1}^*$,
\begin{align} \label{M_x_temporary}
\begin{split}
\frac {\mathcal{M}_x u^{(n,m)}} \lambda  &= 1_{\{L \geq 2\}}\sum_{1 \leq l \leq L-1} \int_{\widetilde{a}_l^*}^{\widetilde{a}_{l-1}^*} \Phi'(p) e^{\Phi(p) (x-y)} f^{(n,m,l)}(y) \diff y +  \int_x^{\widetilde{a}_{L-1}^*}  \Phi'(p) e^{\Phi(p) (x-y)} f^{(n,m,L)}(y) \diff y  \\&+  \int_{\widetilde{a}_{L}^*}^x    \sum_{i \in \mathcal{I}_p} \kappa_{i,p}  e^{-\xi_{i,p} (x-y)} f^{(n,m,L)}(y)  \diff y + 1_{\{L \leq n\}}\sum_{L+1 \leq l \leq n+1} \int_{\widetilde{a}_l^*}^{\widetilde{a}_{l-1}^*}  \sum_{i \in \mathcal{I}_p} \kappa_{i,p}  e^{-\xi_{i,p} (x-y)}f^{(n,m,l)}(y)  \diff y \\
&=1_{\{L \geq 2\}} e^{\Phi(p) x}  \Phi'(p) \sum_{1 \leq l \leq L-1} \varpi_{l}^{(n,m)}(\widetilde{a}_l^*,\widetilde{a}_{l-1}^*, \Phi(p))  +   e^{\Phi(p) x} \Phi'(p)  \varpi_{L}^{(n,m)}(x,\widetilde{a}_{L-1}^*, \Phi(p))   \\&+ \sum_{i \in \mathcal{I}_p}  e^{-\xi_{i,p} x} \kappa_{i,p} \varpi_{L}^{(n,m)}(\widetilde{a}_{L}^*,x, -\xi_{i,p})  +  1_{\{L \leq n\}} \sum_{i \in \mathcal{I}_p}  e^{-\xi_{i,p} x}  \kappa_{i,p}   \sum_{L+1 \leq l \leq n+1} \varpi_{l}^{(n,m)}(\widetilde{a}_{l}^*,\widetilde{a}_{l-1}^*, -\xi_{i,p}),
\end{split}
\end{align}
where
\begin{align}
\varpi_l^{(n,m)}(s, t, q) := \int_{s}^{t}  e^{-q  y}f^{(n,m,l)}(y)  \diff y, \quad s < t. \label{def_varpi_l_n_m}
\end{align}

In particular, for $x > \widetilde{a}_1^*$ (or $L = 1$),
\begin{align}
\begin{split}
\frac {\mathcal{M}_x u^{(n,m)} } \lambda
&= e^{\Phi(p) x} \Phi'(p)  \varpi_{1}^{(n,m)}(x,\infty, \Phi(p))   \\&+ \sum_{i \in \mathcal{I}_p}  e^{-\xi_{i,p} x} \kappa_{i,p} \varpi_{1}^{(n,m)}(\widetilde{a}_{1}^*,x, -\xi_{i,p})  +  \sum_{i \in \mathcal{I}_p}  e^{-\xi_{i,p} x}  \kappa_{i,p}   \sum_{2 \leq l \leq n+1} \varpi_{l}^{(n,m)}(\widetilde{a}_{l}^*,\widetilde{a}_{l-1}^*, -\xi_{i,p}),
\end{split}
\end{align}
while, for $x < \widetilde{a}_n^*$ (or $L = n+1$),
\begin{align}
\begin{split}
\frac {\mathcal{M}_x u^{(n,m)}} \lambda
&= e^{\Phi(p) x}  \Phi'(p) \sum_{1 \leq l \leq n} \varpi_{l}^{(n,m)}(\widetilde{a}_l^*,\widetilde{a}_{l-1}^*, \Phi(p))  +   e^{\Phi(p) x} \Phi'(p)  \varpi_{n+1}^{(n,m)}(x,\widetilde{a}_{n}^*, \Phi(p))   \\&+ \sum_{i \in \mathcal{I}_p}  e^{-\xi_{i,p} x} \kappa_{i,p} \varpi_{n+1}^{(n,m)}(-\infty,x, -\xi_{i,p}).
\end{split}
\end{align}

By repeatedly applying integration by parts, the right-hand side of \eqref{M_x_temporary} can be written in a closed form in the form \eqref{u_form}. {As the computation is straightforward but tedious, we defer the proof of the following proposition and the detailed expressions of the recursive formula \eqref{updating_formula_equation} and  the parameter set to the appendix. }
 
\begin{proposition} \label{proposition_recursion}
Fix $n \geq 1$ and $0 \leq m \leq M-1$ and assume  \eqref{u_form}, \eqref{parameter_set} and \eqref{hypothesis_corner} hold.  Then, with  {$L=L_x$ being the unique integer such that $\tilde{a}_L^*<x<\tilde{a}_{L-1}^*$. We}  
have
%\begin{align*}
%\mathcal{M}_x u^{(n,m)}
%&=e^{\Phi(p) x}  \Phi'(p) \sum_{1 \leq l \leq L} \varpi_{l}(a_l^*,a_{l-1}^*, \Phi(p))  +   e^{\Phi(p) x} \Phi'(p)  \varpi_{L+1}(x,a_{L}^*, \Phi(p))   \\&+ \sum_{i \in \mathcal{I}_p}  e^{-\xi_{i,p} x} \kappa_{i,p} \varpi_{L+1}(a_{L+1}^*,x, -\xi_{i,p})  +  \sum_{i \in \mathcal{I}_p}  e^{-\xi_{i,p} x}  \kappa_{i,p}   \sum_{L+2 \leq l \leq n+1} \varpi_{l}(a_{l}^*,a_{l-1}^*, -\xi_{i,p}),
%\end{align*}
\begin{multline} \label{updating_formula_equation}
\frac {u^{(n,m+1)}(x)} \lambda = \frac {\mathcal{M}_x u^{(n,m)}} \lambda   = A^{(n,m+1,L)} + B^{(n,m+1,L)} e^x  \\ +  \sum_{i \in \mathcal{I}_p} \sum_{h=0}^{I_{n,m+1}} ( C^{(n,m+1,L)}_{i,h} e^{- \xi_{i,p} x} x^h) +  \sum_{h=0}^{I_{n,m+1}} (D^{(n,m+1,L)}_h e^{\Phi(p) x} x^h) + E^{(n,m+1,L)} e^{\Phi(\alpha) x},
\end{multline}
for some parameter set \begin{align}
\Gamma_{n,m+1} := \big[ A^{(n,m+1,l)}, B^{(n,m+1,l)}, \{ C^{(n,m+1,l)}_{i,h}, i \in \mathcal{I}_p, h \geq 0 \}, \{ D^{(n,m+1,l)}_h, h \geq 0 \}, E^{(n,m+1,l)} \big]_{1 \leq l \leq n+1}
\end{align}
that  satisfies \eqref{hypothesis_corner}.
\end{proposition}

\section{Numerical Results} \label{section_numerical_results}
 In this section, we evaluate our method numerically.  For $X$, we use a spectrally negative \lev process of the form \eqref{form_levy_compound_poisson} with a modification so that $(Z_n)$ is independent, identically distributed sequence of the following random variables:
\begin{description}
\item[Case 1] Exponential random variable with parameter $1$,
\item[Case 2] Weibull random variable with parameter  $(2,1)$,
\item[Case 3] The absolute value of the Gaussian (folded normal) with mean zero and variance $1$,
\end{description}
whose respective densities are
\begin{align}
\exp\{-x\}, \quad 2 x \exp \left\{ -  x^2\right\} \quad \textrm{and} \quad \frac 2 {\sqrt{2 \pi}} \exp \left\{ - \frac {x^2} {2} \right\}, \quad x\in(0,\infty).
\end{align}

\textbf{Case 1} is a special case of a phase-type random variable and its scale function can be computed exactly in the form \eqref{scale_function_simple}.  For \textbf{Cases 2} and \textbf{3}, the EM-algorithm is applied to approximate them by phase-type distributions for $d=6$: the fitted phase-type distributions are
\begin{align}
&{\bm T} = \left[ \begin{array}{rrrrrr}   -5.6546  &  0.0000   &      0.0000  &  0.0000  &  0.0000  &  0.0000 \\
    0.6066  & -5.6847 &   0.0000  &  0.0166  &  0.0089  &  5.0526 \\
    0.2156  &  4.3616  & -5.6485  &  0.9162 &   0.1424 &   0.0126 \\
    5.6247 &   0.0000   & 0.0000  & -5.6786  &  0.0000  &  0.0000 \\
    0.0107  &  0.0000  &  0.0000 &   5.7247 &  -5.7420   & 0.0000 \\
    0.0136  &  0.0000 &   0.0000  &  0.0024 &   5.7022 &  -5.7183 \end{array} \right],  \quad {\bm \alpha} =    \left[ \begin{array}{l}
    0.0000 \\
    0.0007 \\
    0.9961 \\
    0.0000 \\
    0.0001 \\
    0.0031  \end{array} \right],
\end{align}
and
\begin{align}
&{\bm T} = \left[ \begin{array}{rrrrrr}-4.0488  &  0.0000   & 0.0000  &  0.0000 &   0.0000 &   0.0000 \\
    0.1320  & -4.0012 &  0.0000  &  0.0455 &   3.7040  & 0.0044 \\
    0.2367  &  0.8595   &-4.2831  &  0.1897   & 0.2918   & 2.3724 \\
    3.1532   & 0.0000   & 0.0000 &  -4.0229  &  0.0000  &  0.0000 \\
    0.2497  &  0.0000  &  0.0000  &  3.7024  & -4.0124  &  0.0000 \\
    0.0434   & 2.1947  &  0.0938  &  0.1704  &  0.1217 &  -4.9612 \end{array} \right],  \quad {\bm \alpha} =    \left[ \begin{array}{r}   0.0052 \\
    0.0659 \\
    0.7446 \\
    0.0398 \\
    0.0043 \\
    0.1403  \end{array} \right],
\end{align}
for \textbf{Cases 2} and \textbf{3}, respectively.  For this phase-type fitting, we use EMpht which is written in C and is publicly available\footnote{Available at http://home.imf.au.dk/asmus/pspapers.html as of March 14, 2014.}. For more details of this method, we refer to \cite{Asmussen_1996}.  {Here, we choose $d = 6$ because, as has been confirmed in  \cite{Egami_Yamazaki_2010_2}, the fitting can be conducted accurately and quickly.
While accuracy tends to increase in $d$, run time increases nonlinearly; hence $d$ cannot be chosen arbitrarily large. As we shall see below, our choice of $d$ attains very small fitting errors.}

For numerical illustration, we set the negative discount rate $\alpha = -0.02$ along with $K = 100$.
For each case of $Z$, we consider the \lev process with common parameters $\rho = 1.5$ and $\sigma = 0.2$ and
choose $\widetilde{c}$ so that $\psi(1) = \alpha  -\gamma$ (i.e. $\exp(-(\alpha - \gamma) t+X_t)_{t \geq 0}$ is a martingale) for our choice of $\gamma$.
Notice that $\widetilde{c}$ (and hence $\E X_1$ as well) decreases as $\gamma$ increases. In the context of stock loans, as discussed in  \cite{Cai2014,xiazhou}, the negative discount rate is the difference of the risk-free rate and the loan rate,    $K$ is the loan amount, and $\gamma$ is the dividend rate.  All the  numerical results given below are generated by MATLAB scripts with double precision on a Windows 7 computer with  an Intel Xeon CPU E$5-2620$, $2.00$GHz, 24.0GB RAM. 
\begin{table}[htbp]
\begin{center}
\begin{tabular}{|c|rcr|rcr|rcr|rcr|}
\hline
 & \multicolumn{6}{c|}{$\gamma = 0.02$} & \multicolumn{6}{c|}{$\gamma = 0.1$}\\
\hline 
 &    \multicolumn{3}{c|}{$M=1$} &  \multicolumn{3}{c|}{$M=3$}&  \multicolumn{3}{c|}{$M=1$} &  \multicolumn{3}{c|}{$M=3$}  \\
\hline
$\xi_{1,p}$   &1.0252&+& 0.0000i&1.5941 &+& 0.0000i&1.0056 &+& 0.0000i & 1.5825 &+& 0.0000i \\
$\xi_{2,p}$   &3.8602 &+& 3.6058i&3.9134 &+& 3.3255i& 3.8296 &+& 3.6319i & 3.8939 &+& 3.3384i\\
$\xi_{3,p}$  & 3.8602 &-& 3.6058i &  3.9134 &-& 3.3255i& 3.8296 &-& 3.6319i&  3.8939 &-& 3.3384i\\
$\xi_{4,p}$  &7.8211 &+& 3.4389i&  7.6518 &+& 3.2454i&7.8398 &+& 3.4933i &7.6613 &+& 3.2799i\\
$\xi_{5,p}$   & 7.8211 &-& 3.4389i&7.6518 &-& 3.2454i&7.8398 &-& 3.4933i& 7.6613 &-& 3.2799i\\
$\xi_{6,p}$  & 9.5837 &+& 0.0000i&9.3632 &+& 0.0000i&9.6386 &+& 0.0000i &9.3983 &+& 0.0000i \\
$\xi_{7,p}$  & 42.040 &+& 0.0000i& 46.026 &+& 0.0000i&  38.4292 &+& 0.0000i&  42.666 &+& 0.0000i\\
  \hline
\end{tabular} \\ 
\textbf{Case2:} Weibull \\
\begin{tabular}{|c|rcr|rcr|rcr|rcr|}
\hline
 & \multicolumn{6}{c|}{$\gamma = 0.02$} & \multicolumn{6}{c|}{$\gamma = 0.1$}\\
\hline 
 &    \multicolumn{3}{c|}{$M=1$} &  \multicolumn{3}{c|}{$M=3$}&  \multicolumn{3}{c|}{$M=1$} &  \multicolumn{3}{c|}{$M=3$}  \\
\hline
$\xi_{1,p}$   &0.9842 &+& 0.0000i&1.4669 &+& 0.0000i&0.9674 &+& 0.0000i& 1.4583 &+& 0.0000i \\
$\xi_{2,p}$   & 3.2497 &+& 2.3023i&3.2876 &+& 2.0887i& 3.2331 &+& 2.3200i & 3.2784 &+& 2.0976i\\
$\xi_{3,p}$  & 3.2497 &-& 2.3023i & 3.2876 &-& 2.0887i&   3.2331 &-& 2.3200i&  3.2784 &-& 2.0976i\\
$\xi_{4,p}$  &5.5298 &+& 1.6297i&   5.4233 &+& 1.5437i&  5.5425 &+& 1.6464i & 5.4300 &+& 1.5543i\\
$\xi_{5,p}$   &   5.5298 &-& 1.6297i&5.4233 &-& 1.5437i& 5.5425 &-& 1.6464i&  5.4300 &-& 1.5543i\\
$\xi_{6,p}$  & 6.4520 &+& 0.0000i&6.2947 &+& 0.0000i&6.4805 &+& 0.0000i &6.3103 &+& 0.0000i \\
$\xi_{7,p}$  &  37.565 &+& 0.0000i& 41.862 &+& 0.0000i&   34.049 &+& 0.0000i& 
 38.617 &+& 0.0000i\\
  \hline
\end{tabular} \\ 
\textbf{Case 3:} Folded Normal
\end{center}
\caption{{Values of $\xi_{i,p}$  for $M=1,3$ and $\gamma = 0.02, 0.1$ (listed in ascending order). We can confirm that these values are all distinct. Because $X$ has a Brownian motion component, $|\mathcal{I}_p| = d + 1 = 7.$}} \label{table_xi}
\end{table}

\begin{table}[htbp]
\begin{tabular}{c||c|c||c|c|}
      & \multicolumn{2}{c||}{randomization}  & \multicolumn{2}{c|}{simulation}  \\ \hline
$M$              & value & time  & value & time  \\
\hline
$1$&$1823.65$&$0.306+0.008$&$1823.89(1821.61,1826.17)$&$146.712$\\
$2$&$1824.27$&$0.242+0.024$&$1824.15(1822.03,1826.28)$&$152.969$\\
$3$&$1824.51$&$0.245+0.066$&$1824.58(1822.51,1826.66)$&$157.001$\\
$4$&$1824.64$&$0.240+0.146$&$1823.69(1821.71,1825.68)$&$162.082$\\
$5$&$1824.72$&$0.238+0.271$&$1825.10(1823.00,1827.19)$&$167.068$\\
\hline
$10$&$1824.88$&$0.273+2.008$&$1823.11(1821.01,1825.20)$&$186.893$\\
\hline
const&\multicolumn{2}{c||}{N/A}&$1823.90(1821.80,1826.00)$&$141.833$
\end{tabular} \\ \vspace{0.3cm}
\textbf{Case 1}: Exponential \\ \vspace{0.5cm}
\begin{tabular}{c||c|c||c|c|}
      & \multicolumn{2}{c||}{randomization}  & \multicolumn{2}{c|}{simulation}  \\ \hline
$M$              & value & time  & value & time  \\
\hline
$1$&$1665.62$&$0.604+0.036$&$1665.68(1663.64,1667.73)$&$395.063$\\
$2$&$1666.12$&$0.455+0.201$&$1663.54(1661.63,1665.44)$&$393.945$\\
$3$&$1666.32$&$0.374+0.599$&$1666.02(1664.07,1667.97)$&$404.590$\\
$4$&$1666.42$&$0.524+1.350$&$1664.73(1662.84,1666.62)$&$403.772$\\
$5$&$1666.49$&$0.534+2.520$&$1665.47(1663.59,1667.34)$&$411.322$\\
\hline
$10$&$1666.58$&$0.403+18.74$&$1667.07(1665.08,1669.06)$&$426.856$\\
\hline
const&\multicolumn{2}{c||}{N/A}&$1666.61(1664.75,1668.47)$&$386.804$
\end{tabular} \\ \vspace{0.3cm}
\textbf{Case 2}: Weibull \\ \vspace{0.5cm}
\begin{tabular}{c||c|c||c|c|}
      & \multicolumn{2}{c||}{randomization}  & \multicolumn{2}{c|}{simulation}  \\ \hline
$M$              & value & time  & value & time  \\
\hline
$1$&$1482.88$&$0.594+0.036$&$1486.05(1484.40,1487.69)$&$141.351$\\
$2$&$1483.35$&$0.380+0.232$&$1484.31(1482.76,1485.86)$&$147.091$\\
$3$&$1483.53$&$0.389+0.584$&$1484.30(1482.72,1485.89)$&$152.59$\\
$4$&$1483.63$&$0.379+1.302$&$1484.06(1482.45,1485.67)$&$156.25$\\
$5$&$1483.69$&$0.382+2.471$&$1485.29(1483.69,1486.89)$&$161.923$\\
\hline
$10$&$1483.80$&$0.329+18.64$&$1485.24(1483.67,1486.81)$&$184.049$ \\
\hline
const&\multicolumn{2}{c||}{N/A}&$1485.35(1483.78,1486.91)$&$137.375$

\end{tabular} \\ \vspace{0.3cm}
\textbf{Case 3}: Folded Normal

\caption{Comparison between results under randomization and simulation for $\gamma = 0.02$. {The comparison is done for each Erlang shape parameter $M$; in the bottom row (labeled const), the approximated values under simulation for the constant $\delta = 0.5$ case are given.  The listed values under simulation are the mean and $95\%$ confidence interval. The computation times  (in seconds) for randomization are given as a sum of the time spent for steps (i) and (ii) . }} \label{table_randomization_simulation_002}
\end{table}

\begin{table}[htbp]
\begin{tabular}{c||c|c||c|c|}
      & \multicolumn{2}{c||}{randomization}  & \multicolumn{2}{c|}{simulation}  \\ \hline
$M$              & value & time  & value & time  \\
\hline
$1$&$323.83$&$0.360+0.008$&$323.85(323.44,324.27)$&$146.511$\\
$2$&$324.33$&$0.235+0.023$&$324.10(323.69,324.51)$&$150.866$\\
$3$&$324.54$&$0.231+0.064$&$324.23(323.82,324.64)$&$157.942$\\
$4$&$324.65$&$0.235+0.143$&$324.13(323.74,324.53)$&$162.298$\\
$5$&$324.72$&$0.237+0.273$&$324.51(324.10,324.91)$&$168.973$\\
\hline
$10$&$324.87$&$0.271+2.013$&$324.48(324.06,324.90)$&$184.726$\\
\hline
const&\multicolumn{2}{c||}{N/A}&$324.97(324.56,325.37)$&$142.774$
\end{tabular} \\ \vspace{0.3cm}
\textbf{Case 1}: Exponential \\ \vspace{0.5cm}
\begin{tabular}{c||c|c||c|c|}
      & \multicolumn{2}{c||}{randomization}  & \multicolumn{2}{c|}{simulation}  \\ \hline
$M$              & value & time  & value & time  \\
\hline
$1$&$303.13$&$0.307+0.036$&$302.36(301.94,302.78)$&$389.083$\\
$2$&$303.54$&$0.271+0.202$&$303.44(303.04,303.85)$&$396.285$\\
$3$&$303.72$&$0.588+0.679$&$303.63(303.23,304.03)$&$401.650$\\
$4$&$303.81$&$0.377+1.458$&$303.48(303.09,303.86)$&$402.040$\\
$5$&$303.87$&$0.401+2.504$&$303.65(303.24,304.06)$&$404.430$\\
\hline
$10$&$304.00$&$0.335+18.69$&$303.87(303.48,304.25)$&$428.611$\\
\hline
const&\multicolumn{2}{c||}{N/A}&$303.98(303.60,304.37)$&$385.145$
\end{tabular} \\ \vspace{0.3cm}
\textbf{Case 2}: Weibull \\ \vspace{0.5cm}
\begin{tabular}{c||c|c||c|c|}
      & \multicolumn{2}{c||}{randomization}  & \multicolumn{2}{c|}{simulation}  \\ \hline
$M$              & value & time  & value & time  \\
\hline
$1$&$265.46$&$0.876+0.037$&$265.67(265.33,266.00)$&$144.601$\\
$2$&$265.85$&$0.380+0.235$&$265.90(265.57,266.23)$&$150.759$\\
$3$&$266.01$&$0.386+0.583$&$266.37(266.04,266.70)$&$155.104$\\
$4$&$266.10$&$0.341+1.311$&$266.49(266.17,266.81)$&$158.389$\\
$5$&$266.15$&$0.388+2.459$&$266.69(266.37,267.01)$&$162.124$\\
\hline
$10$&$266.28$&$0.336+18.75$&$266.50(266.18,266.81)$&$182.452$\\
\hline
const&\multicolumn{2}{c||}{N/A}&$266.86(266.55,267.17)$&$139.149$
\end{tabular} \\ \vspace{0.3cm}
\textbf{Case 3}: Folded Normal
\caption{Comparison between results under randomization and simulation for $\gamma = 0.1$.} \label{table_randomization_simulation_01}
\end{table}

\subsection{One-stage randomization}
We first analyze the accuracy and computation time of our randomization algorithm by considering the expectation, for $\delta = 0.5$,
\begin{align}
\E_x [e^{-\alpha \delta } v^{(1)} (X_\delta)], \label{def_expectation_delta}
\end{align}
where $v^{(1)}$ is analytically given as in \eqref{u_1_0}.  In order to do so, we evaluate the approximations by our algorithm in comparison to the simulated results.  More specifically, we first compute, for $M=1, \ldots, 5,10$,  the approximations by these two methods to
\begin{align}
u^{(1,M)}(x) = \E_x [e^{-\alpha \eta(M, M/\delta) } v^{(1)} (X_{\eta(M, M/\delta)})], \label{def_expectation_delta_erlang}
\end{align}
and then approximate the constant $\delta$ case \eqref{def_expectation_delta} by simulation with a starting point $x= \widetilde{a}_1^*$.  This enables us to analyze the approximation errors of our randomization algorithm for the Erlang case \eqref{def_expectation_delta_erlang}, and also analyze how large $M$ needs to be to acquire accurate approximations for the constant $\delta$ case \eqref{def_expectation_delta}.

Following the arguments in the previous section, our computation involves   two main steps: (i) computing the roots of  $\psi(\cdot)=p$, and (ii) computing recursively the parameter set $\Gamma$  in \eqref{parameter_set}.  The root-finding procedure  is conducted   by MATLAB built-in function \texttt{solve()}.  {In Table \ref{table_xi}, we give sample values of $\xi_{i,p}$; here we can confirm that these values are all distinct.} The step (ii) can be done efficiently by applying, for $M$ times, Inductive Step II in the previous section.  As for the simulated results, we compute this via Monte Carlo simulation based on $1$ million sample paths, where the  Brownian motions are approximated by random walks with time step $\Delta t = \hat{T}/ 100$ for each inter-arrival time $\hat{T}$ between jumps.

Tables \ref{table_randomization_simulation_002} and \ref{table_randomization_simulation_01} summarize the results for $\gamma = 0.02$ and $0.1$, respectively.  The functions $u^{(1,M)}$'s, as in \eqref{def_expectation_delta_erlang},  obtained from  the analytic recursive formula and simulation   are listed for $M=1, \ldots, 5, 10$, along with the constant $\delta$ case   computed by simulation presented in the bottom row.  We also report the   computation times (in seconds). The times that correspond to  the analytic formula are given as a sum of the time spent for steps (i) and (ii).  For the values under simulation, we give  the  mean and $95\%$ confidence interval for each case.

The simulated results are subject to some errors arising from the  discretization of Brownian motions, but they are useful as a benchmark.   These discretization errors are confirmed to be minimal in view of the comparison between these two methods for \textbf{Case 1}. Recall that the numerical results for  \textbf{Case 1} by the randomization algorithm are  \emph{exact} in the sense that there is no approximation error from  fitting the scale function.   Based on this observation, we can also infer from the results on \textbf{Cases 2} and \textbf{3} that the associated fitting errors of the scale function are also minimal.  This suggests the practicability of the use of the phase-type distribution as an approximation for a general \lev process.

 As $M$ increases from $1$ to $10$,  the approximate value function \eqref{def_expectation_delta_erlang}  increases monotonically  and  approaches the simulated value for \eqref{def_expectation_delta} associated with the constant $\delta$ case.  In fact,  the exponential refraction time case (i.e.\ $M=1$) already gives a reasonable approximation.

In terms of the computation time, the randomization method is significantly faster   than  simulation. Note also that, for the randomization method, this computation is required only once to obtain the whole shape of the value function.  The simulation method, on the other hand, is unfortunately not practical; it takes several minutes to attain this accuracy for a particular point of $x$.  Recall that in our multiple stopping problem, we need to know the whole shape to conduct backward induction. If the simulation method is applied, one needs to compute for arbitrarily large number of starting points $x$. However, this is computationally infeasible.

While the randomization method runs instantaneously when $M$ is small, we observe that the computation time increases nonlinearly in $M$.  It also depends on the number of phases; \textbf{Case 1} (with $1$ phase) runs faster than \textbf{Cases 2} and \textbf{3} (with $6$ phases).  This suggests one limitation of the randomization algorithm that the value of $M$ and the number of phases $d$ cannot be chosen arbitrarily large.  However, as we already see in Tables \ref{table_randomization_simulation_002}  and \ref{table_randomization_simulation_01}, the approximate value function stabilizes  even for small  $M$ {and our choice of $d$}.

\subsection{Multiple-stage case}  We now move on to the multiple-stage case.
Using our randomization algorithm, the approximate value functions $\widetilde{v}^{(1)}, \ldots, \widetilde{v}^{(5)}$ are computed for Erlang shape parameters $M =1,3$ and are shown in  Figures \ref{figure_multiple_002} and \ref{figure_multiple_01}, respectively, for $\gamma = 0.02$ and $0.1$. The threshold levels  $\widetilde{a}_1^*, \ldots, \widetilde{a}_5^*$ (circles) are marked on the approximate value function curves.   In particular, the top curve corresponds to the approximate value function $\widetilde{v}^{(5)}$.   As expected, the thresholds are all above the strike $K=100$ and they admit the ordering $\widetilde{a}_{n+1}^* < \widetilde{a}_{n}^*$.  This is consistent with Remark \ref{remark_monotonicity}.

Recall that  the process $\exp(-(\alpha - \gamma)t + X_t )_{t \geq 0}$ is a martingale under the given parameters.  Hence, for a small value of $\gamma$, the value function is close to linear in $\exp(x)$. On the other hand, as $\gamma$ increases, it appears to be more convex.   Moreover, the function $\widetilde{v}^{(N)}$ decreases as $\gamma$ increases because $\gamma$ reduces the drift of  $X$.  As in the single stopping  case, the difference between the value functions for $M=$ 1 and 3  is close to invisible. This suggests that these are reasonable approximations for the constant $\delta$ case.   On the other hand,  the  optimal threshold levels show {non-negligible} difference between the cases $M=$ 1 and 3.

\begin{figure}[htbp]
\begin{center}
\begin{minipage}{1.0\textwidth}
\centering
\begin{tabular}{cc}
 \includegraphics[scale=0.55]{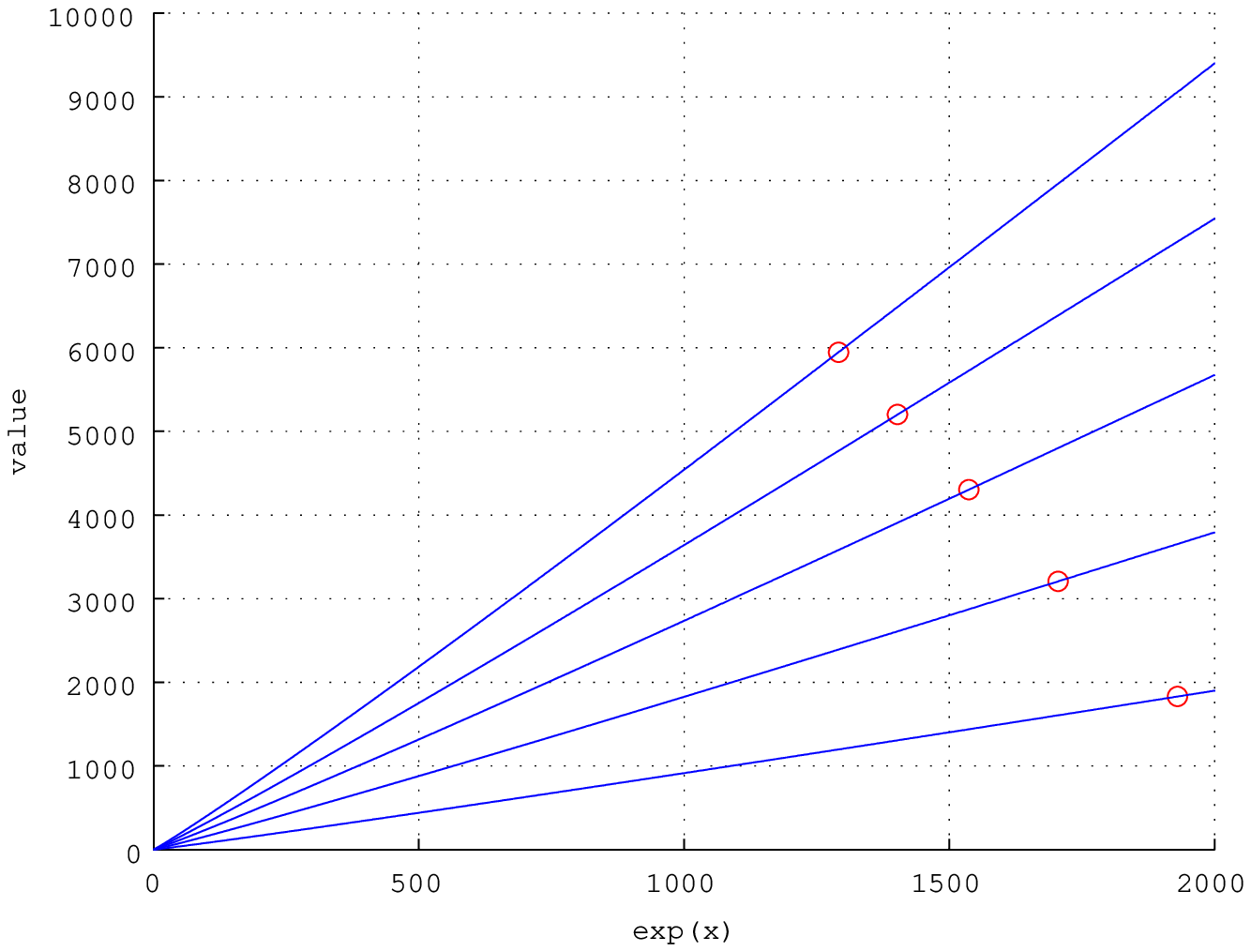} & \includegraphics[scale=0.55]{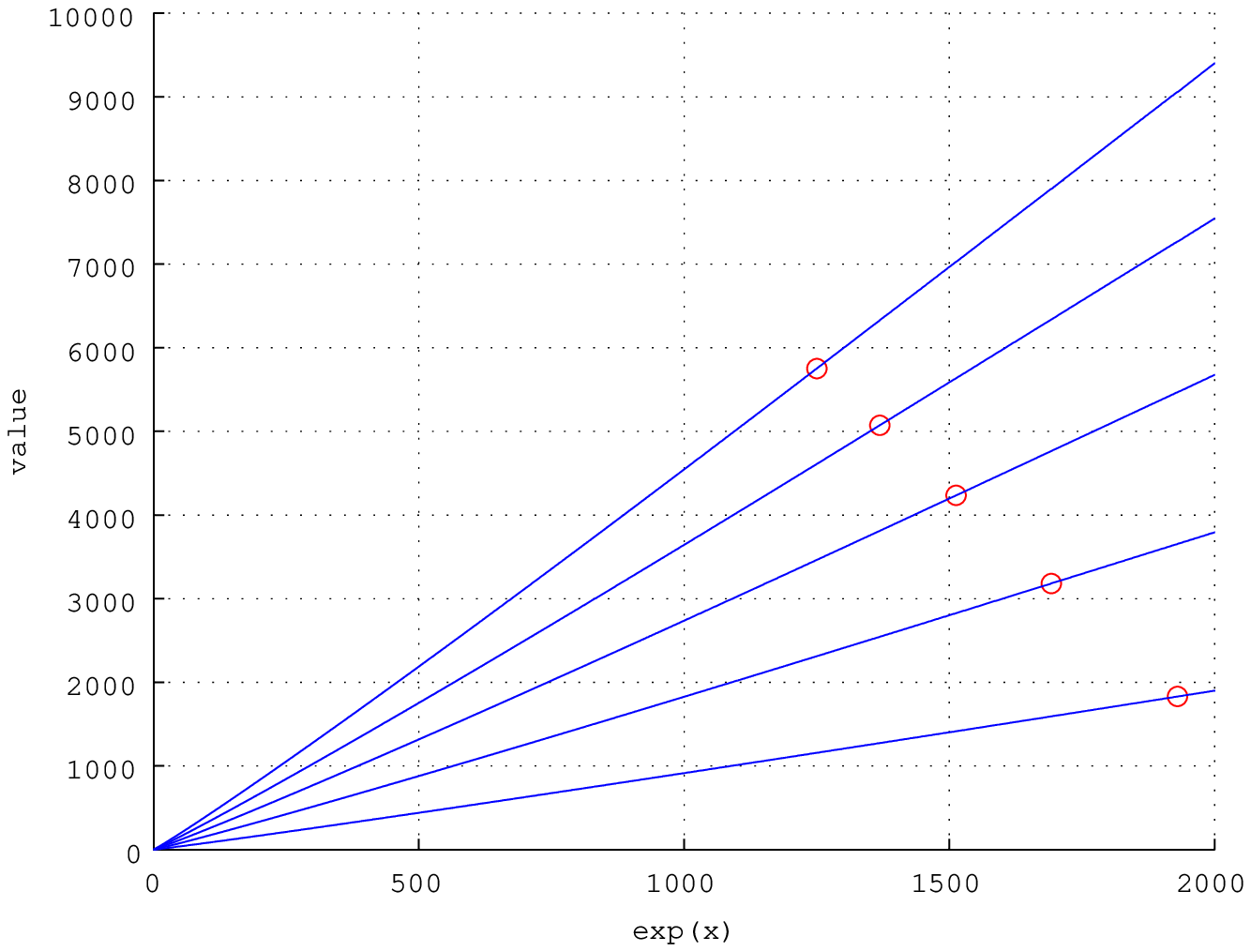}  \\
\textbf{Case 1} (Exponential) with $M=1$ & \textbf{Case 1} (Exponential)   with $M=3$ \vspace{0.3cm} \\
  \includegraphics[scale=0.55]{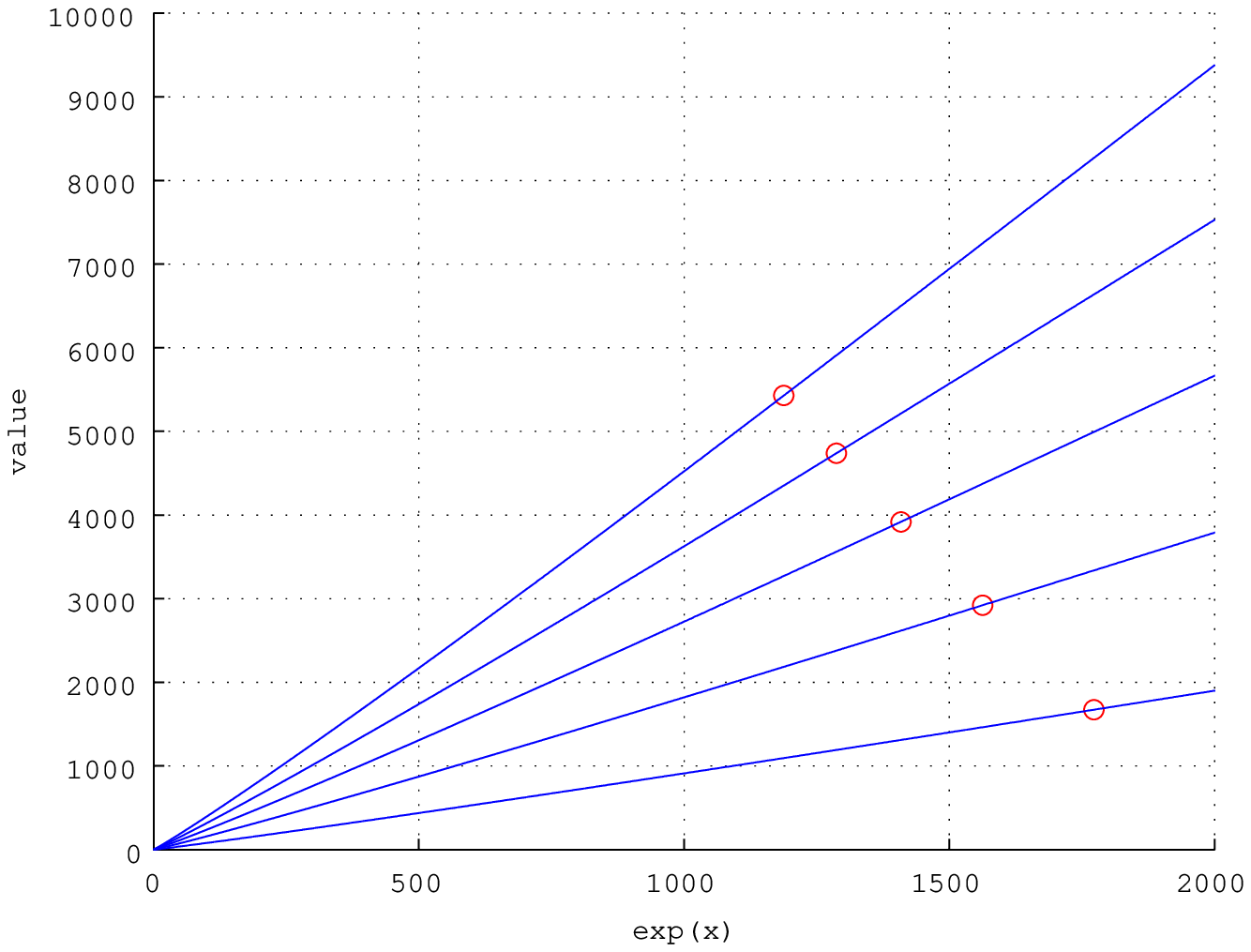} & \includegraphics[scale=0.55]{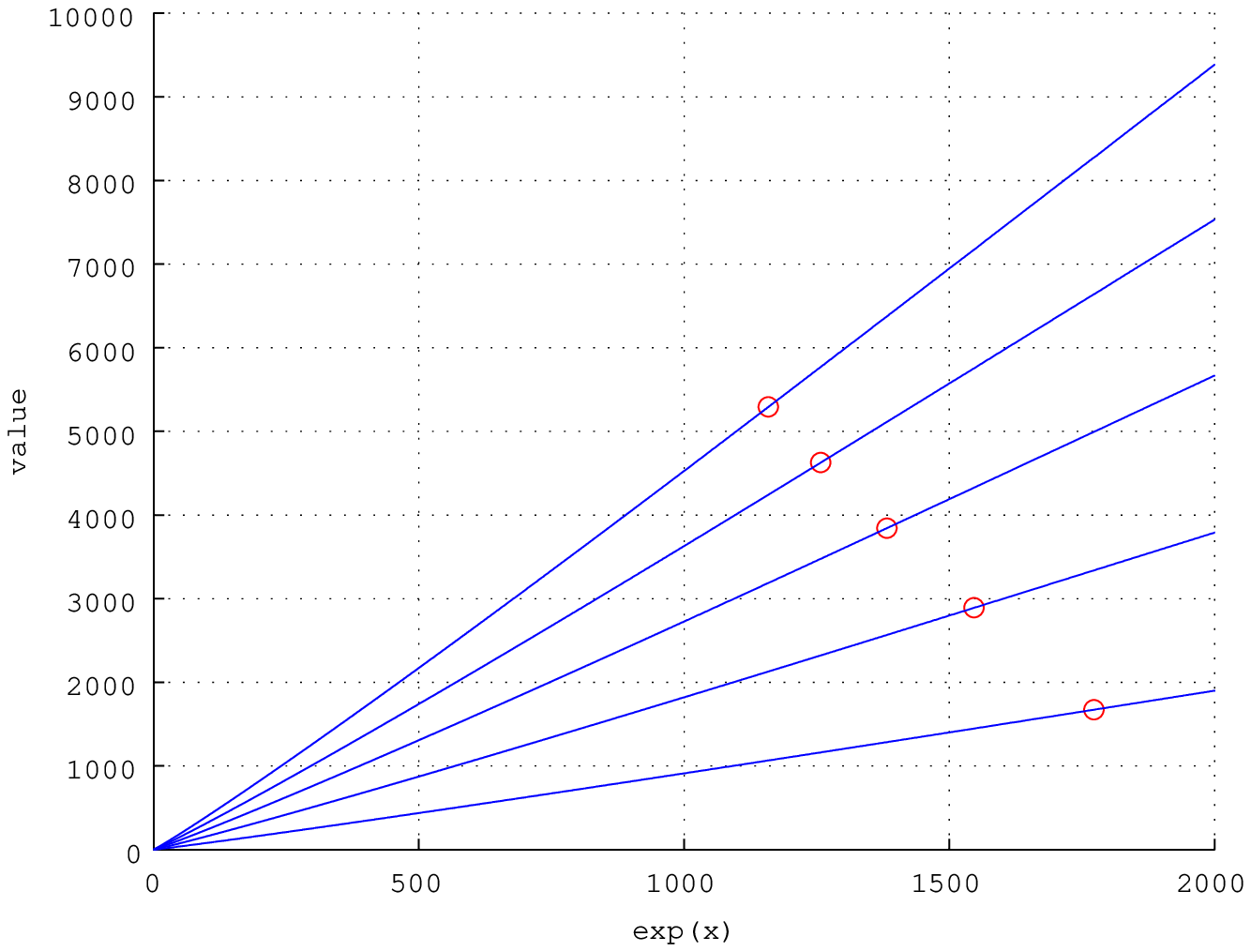}  \\
\textbf{Case 2} (Weibull)  with $M=1$ &  \textbf{Case 2} (Weibull) with $M=3$ \vspace{0.3cm} \\
   \includegraphics[scale=0.55]{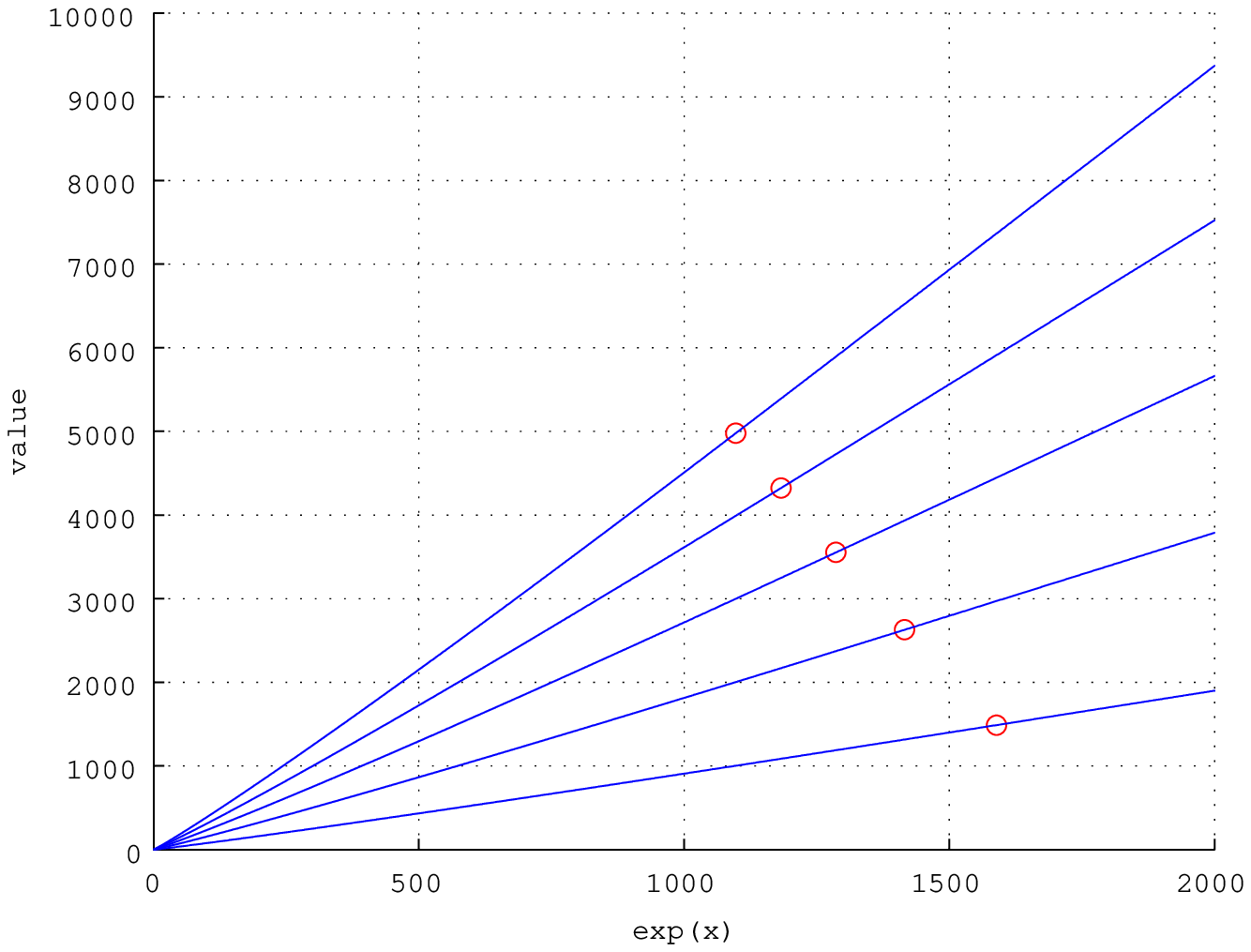} & \includegraphics[scale=0.55]{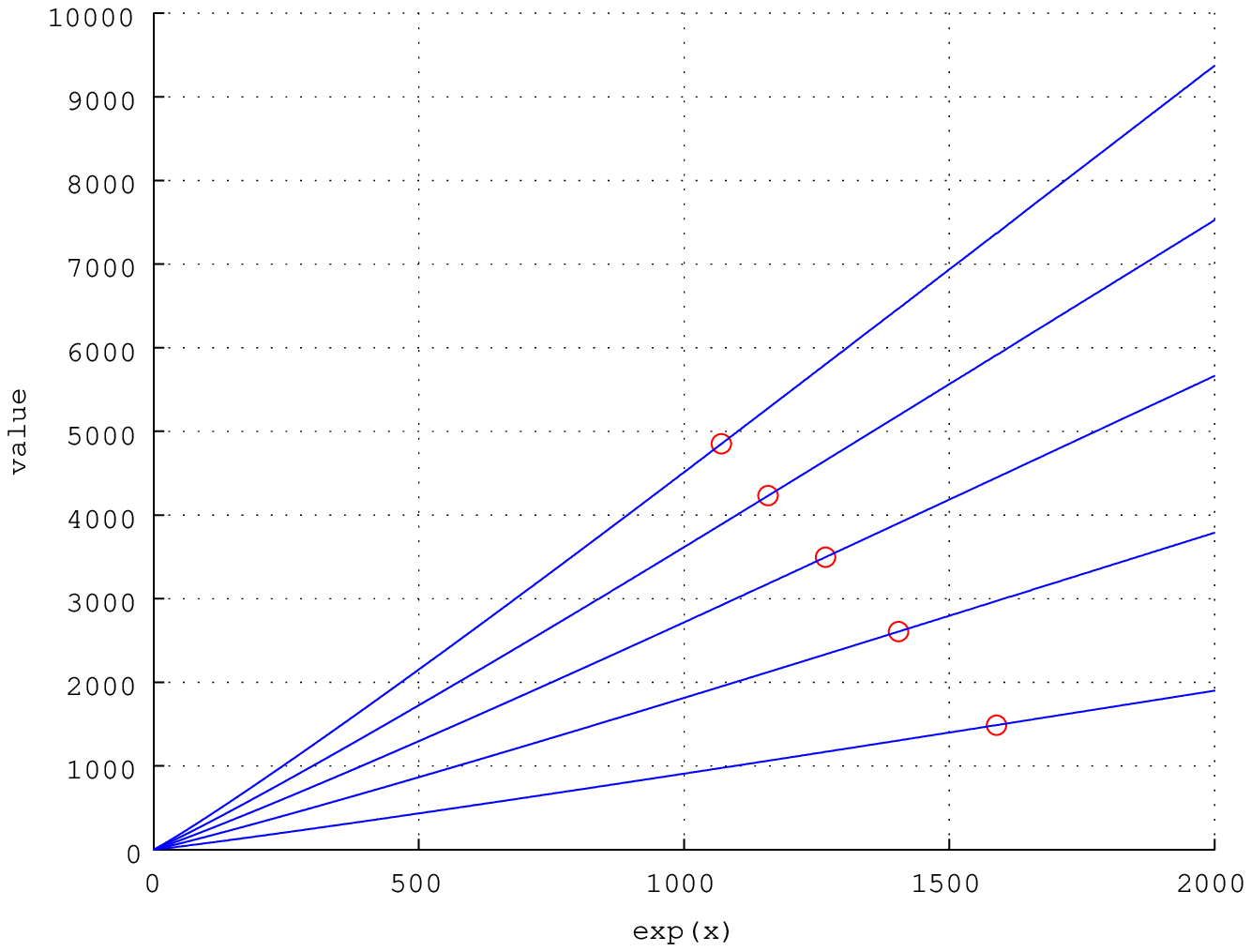}  \\
\textbf{Case 3} (Folded Normal) with $M=1$ &  \textbf{Case 3} (Folded Normal) with $M=3$\vspace{0.3cm} \\
\end{tabular}
\end{minipage}
\caption{The approximate value functions when $\gamma = 0.02$ {with  threshold levels  $\widetilde{a}_1^*, \ldots, \widetilde{a}_5^*$ (circles) marked on the approximate value function curves. The values are monotone in the number of stages (the top curve corresponds to the approximate value function $\widetilde{v}^{(5)}$).}} \label{figure_multiple_002}
\end{center}
\end{figure}

\begin{figure}[htbp]
\begin{center}
\begin{minipage}{1.0\textwidth}
\centering
\begin{tabular}{cc}
 \includegraphics[scale=0.55]{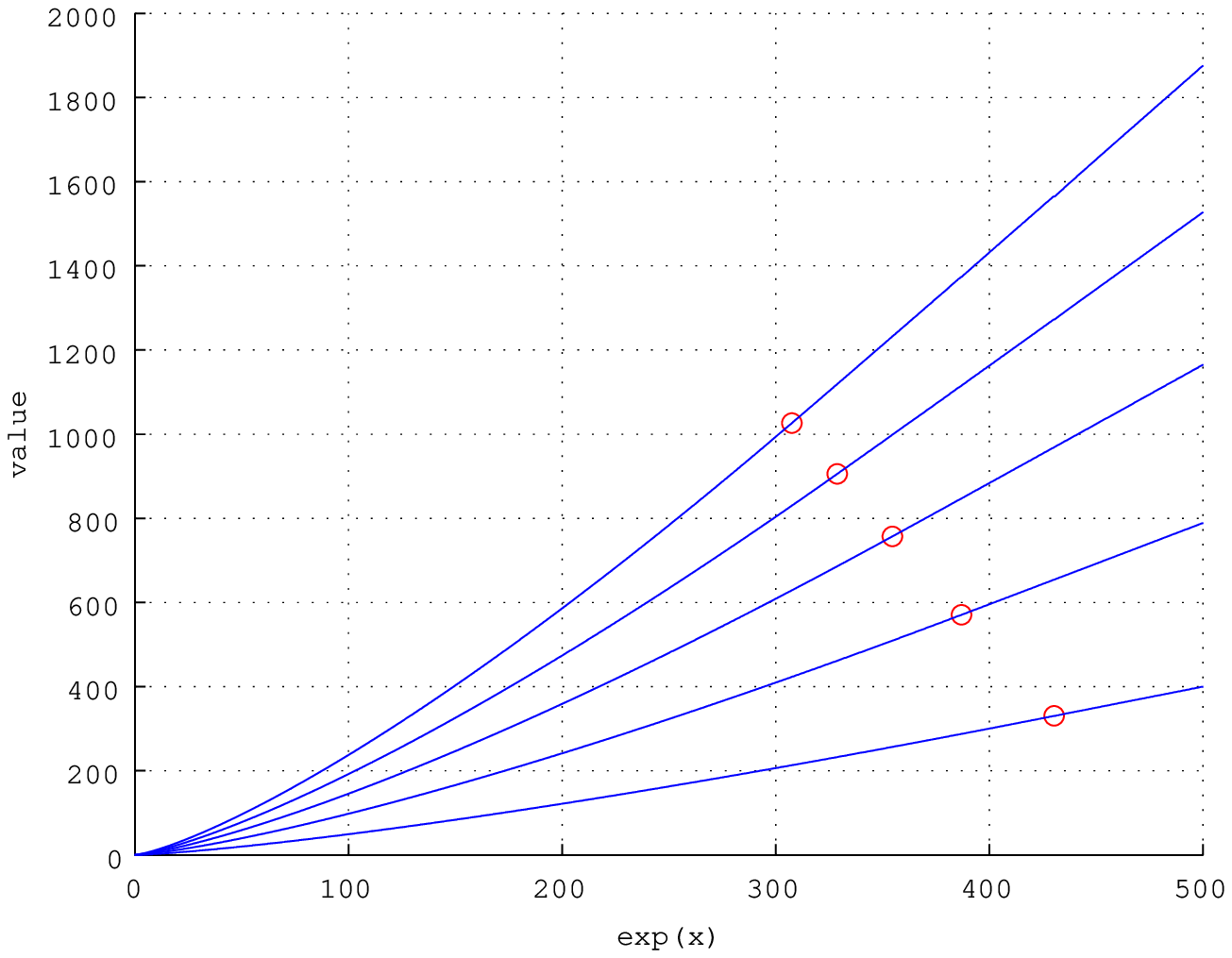} & \includegraphics[scale=0.55]{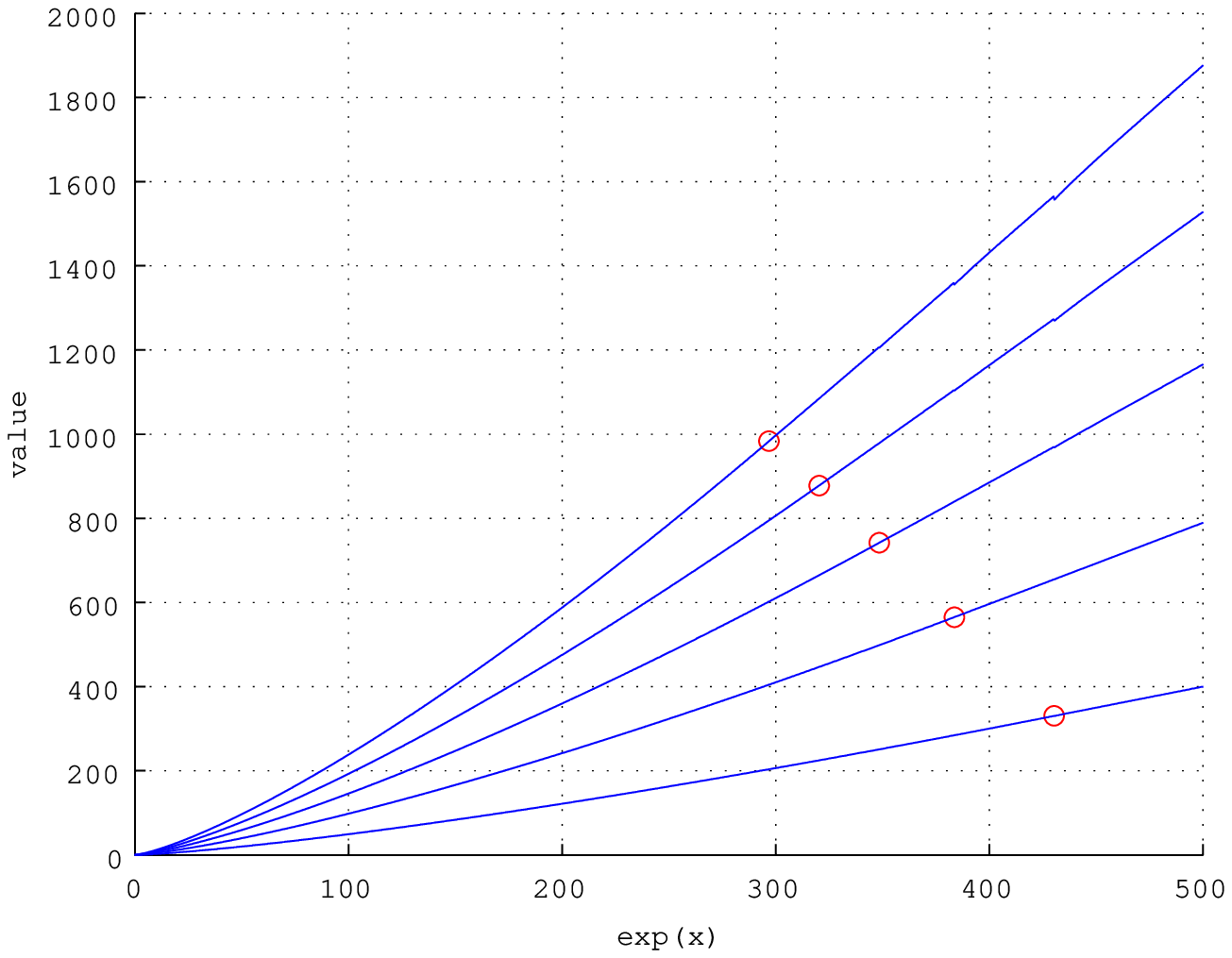}  \\
\textbf{Case 1}  (Exponential) with $M=1$ & \textbf{Case 1}  (Exponential) with $M=3$ \vspace{0.3cm} \\
  \includegraphics[scale=0.55]{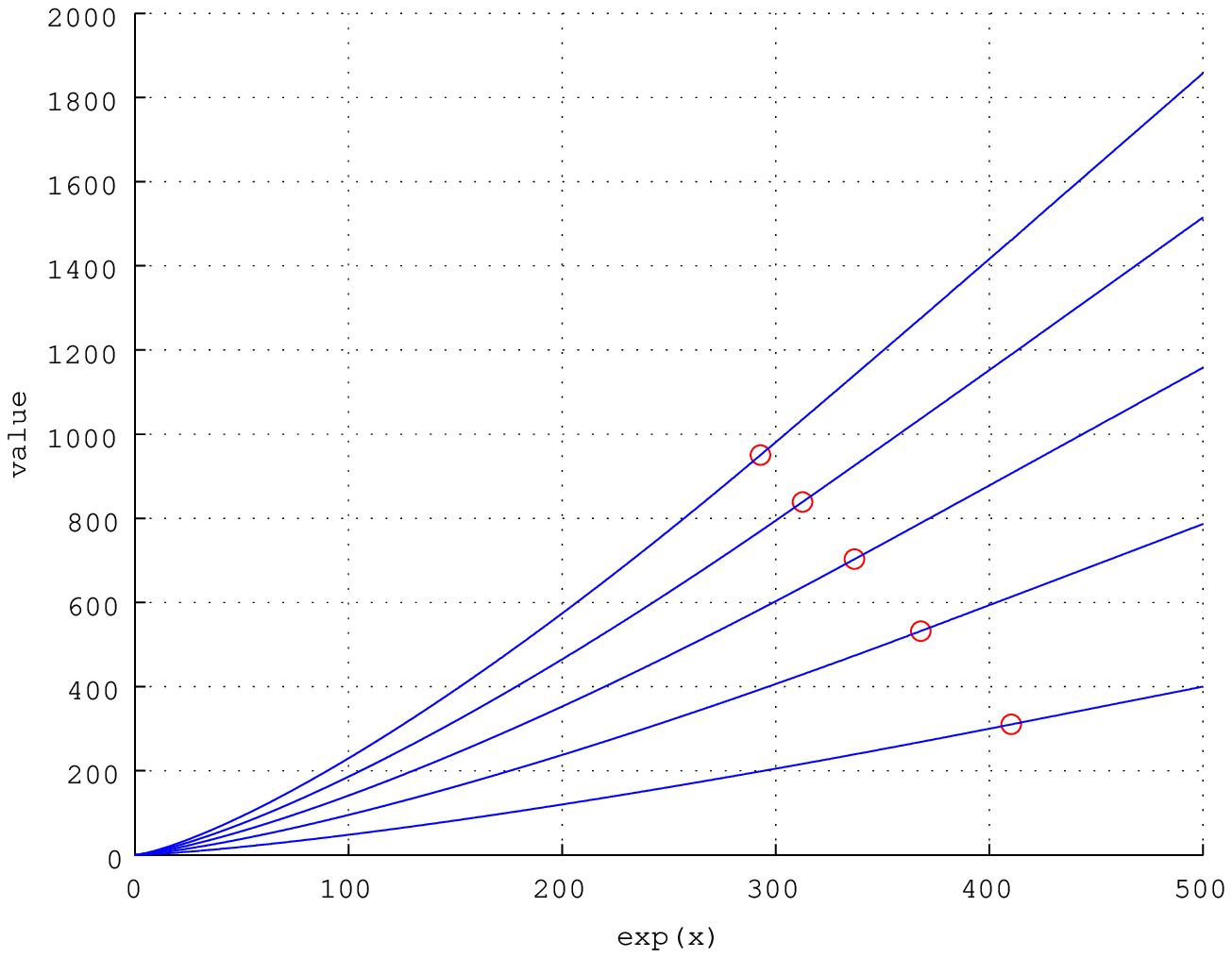} & \includegraphics[scale=0.55]{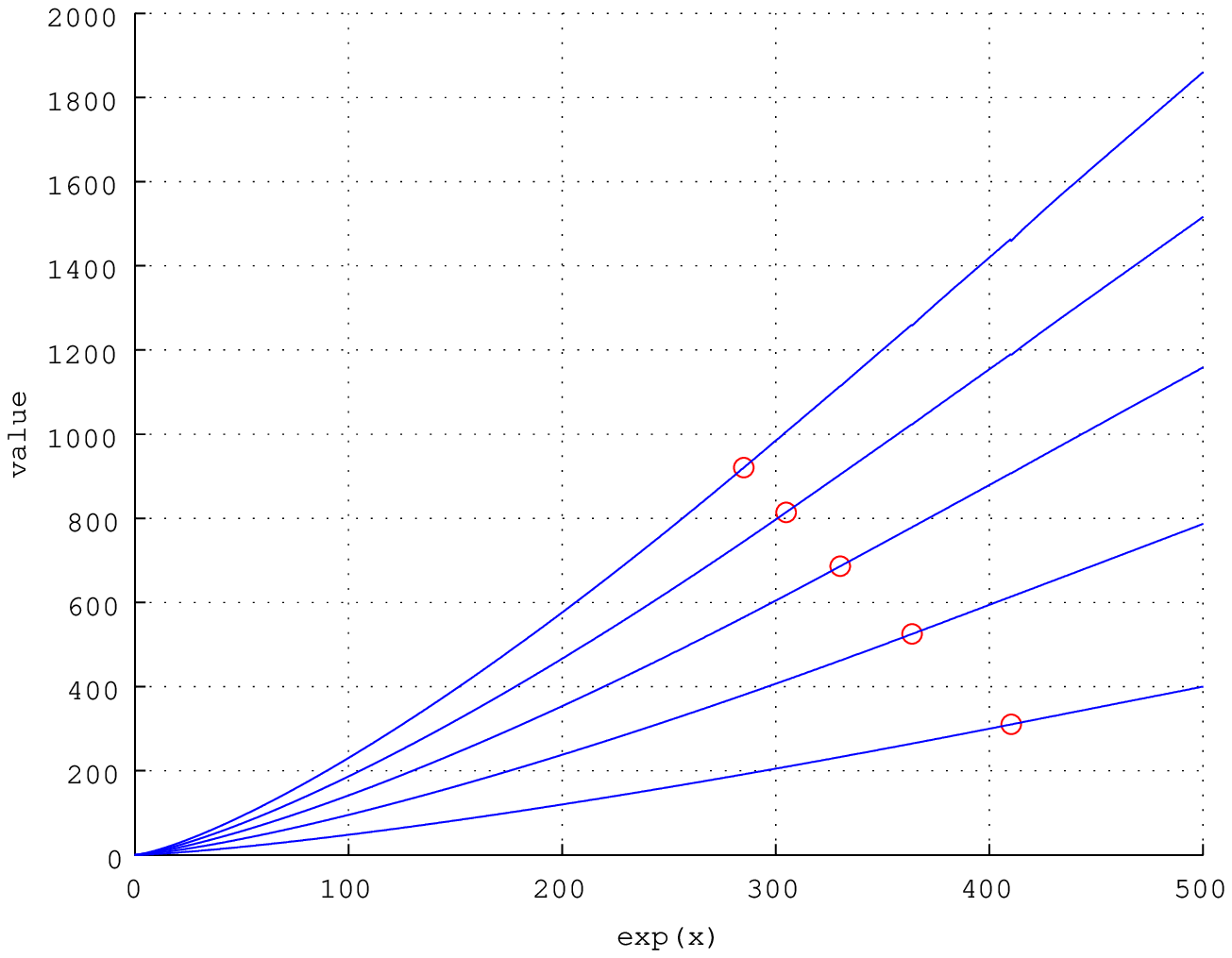}  \\
\textbf{Case 2}   (Weibull) with $M=1$ &  \textbf{Case 2}  (Weibull) with $M=3$ \vspace{0.3cm} \\
   \includegraphics[scale=0.55]{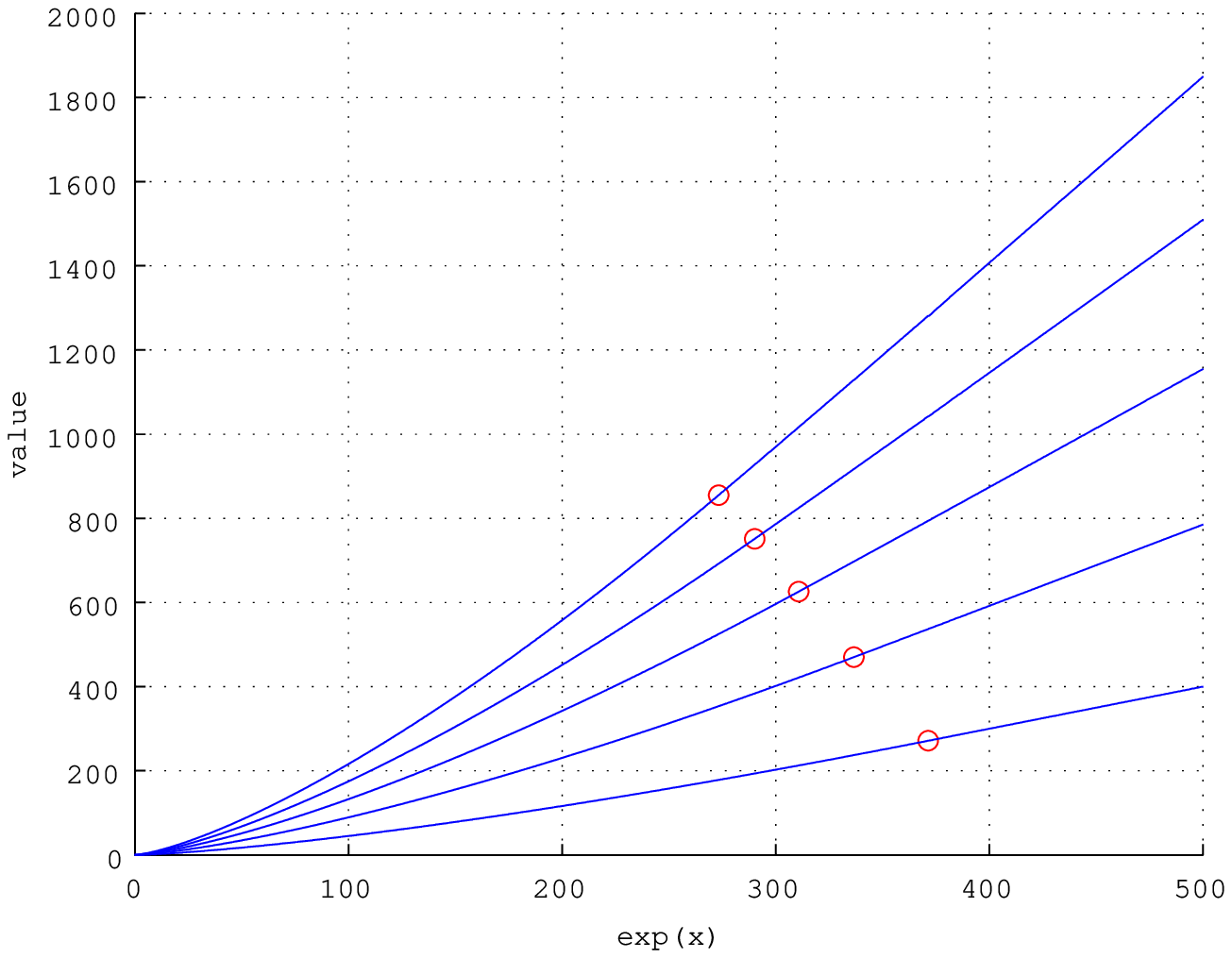} & \includegraphics[scale=0.55]{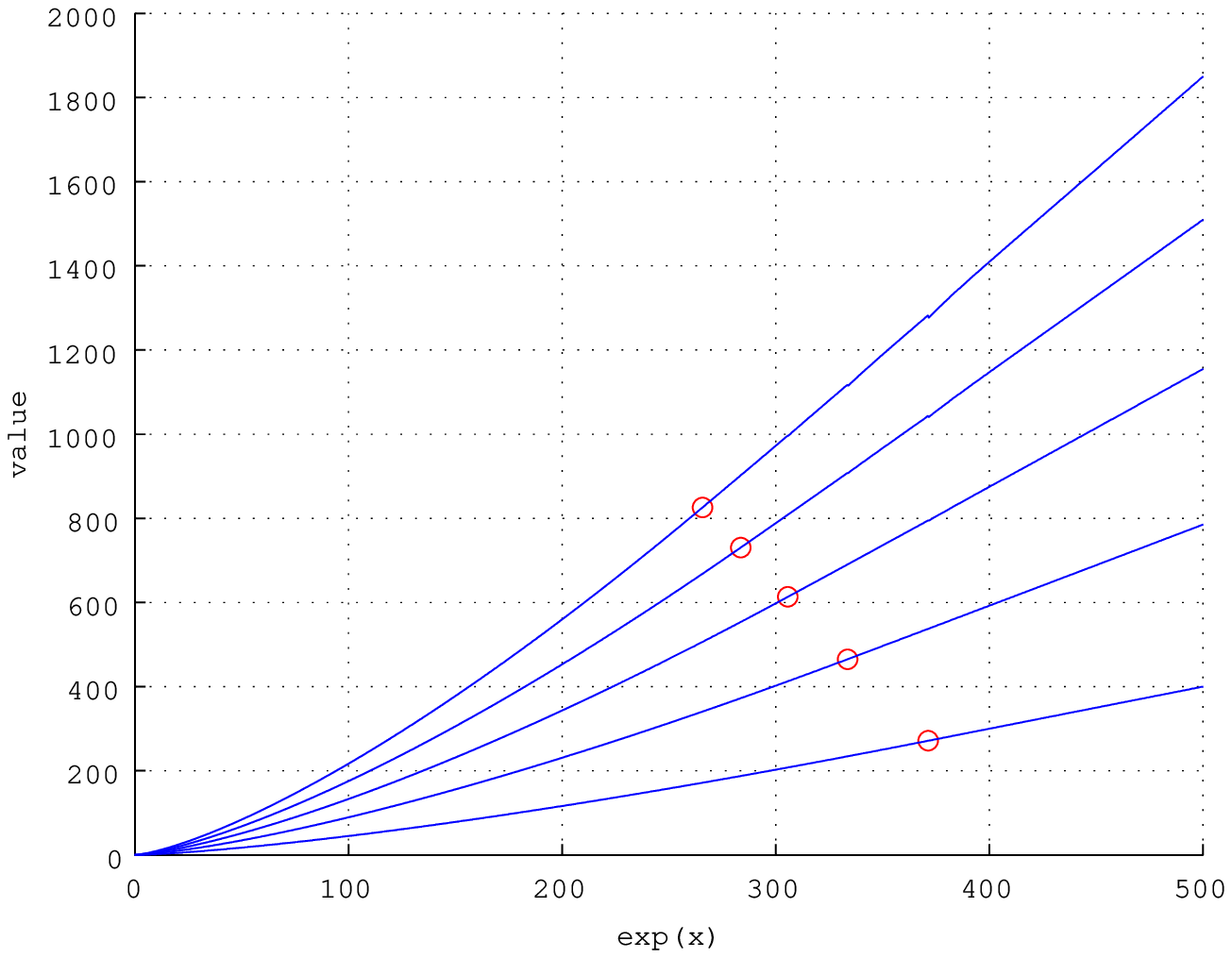}  \\
\textbf{Case 3}  (Folded Normal) with $M=1$ &  \textbf{Case 3}  (Folded Normal) with $M=3$\vspace{0.3cm} \\
\end{tabular}
\end{minipage}
\caption{The approximate value functions when $\gamma = 0.1$. } \label{figure_multiple_01}
\end{center}
\end{figure}

\subsection{Dependence on  $N$ and $M$}
In Figure \ref{figure_convergence_threshold}, we show the threshold levels with respect to the number of stages $N$ and to the Erlang shape parameter $M$ based on  \textbf{Case 3} with $\gamma = 0.1$.   On the left panel, we plot  $\widetilde{a}^*_1, \ldots, \widetilde{a}^*_4$ for fixed $M=1,\ldots, 4$.  Note that the first threshold  $\widetilde{a}^*_1$  is independent of $M$.    In the example on  the right panel, we plot  the threshold  $\widetilde{a}^*_2$ over  $M=1, \ldots, 10$.   As $M$ increases from $1$ to $10$,  the threshold first decreases relatively fast within the narrow range $(5.81, 5.82)$ for $M = 1, 2,3$, and then flattens toward the value 5.805 for larger $M$. Between $M=$ 9 and 10, the difference is well less than 0.001.

  Figure \ref{figure_convergence} illustrates  $\widetilde{v}^{(N)}$ for $N = 1, \ldots, 15$ with  $\gamma=0.05$ and  $M=1$. With more remaining exercise opportunities (large $N$),  the function $\widetilde{v}^{(N)}$ is higher and the optimal threshold for the previous exercise is lower.  We observe that the distance between  successive   optimal thresholds (marked by circles)   reduces  as the number of remaining exercises increases (see e.g. the top value function curve marked with 15 circles).

\begin{figure}[htbp]
\begin{center}
\begin{minipage}{1.0\textwidth}
\centering
\begin{tabular}{cc}
 \includegraphics[scale=0.58]{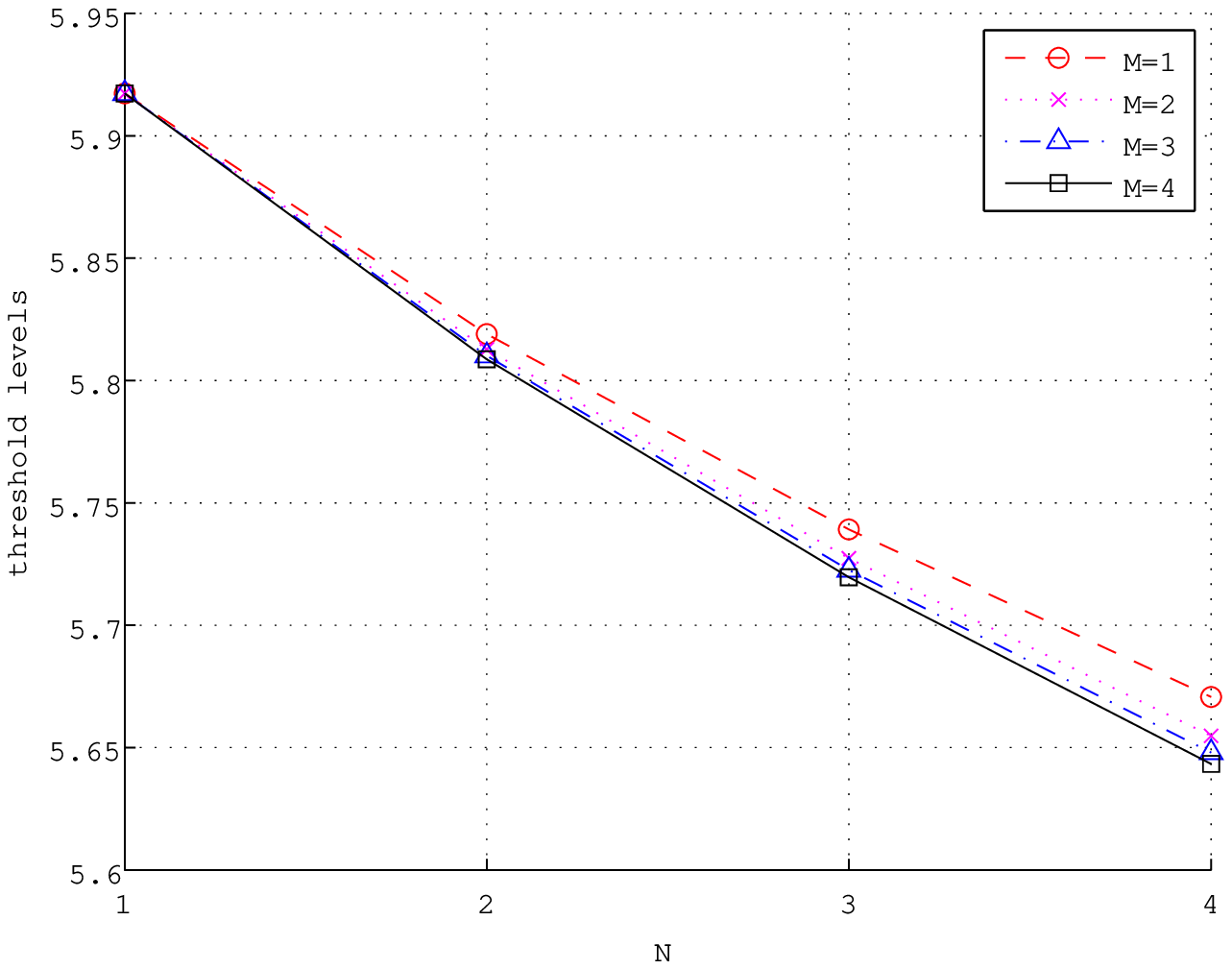} & \includegraphics[scale=0.58]{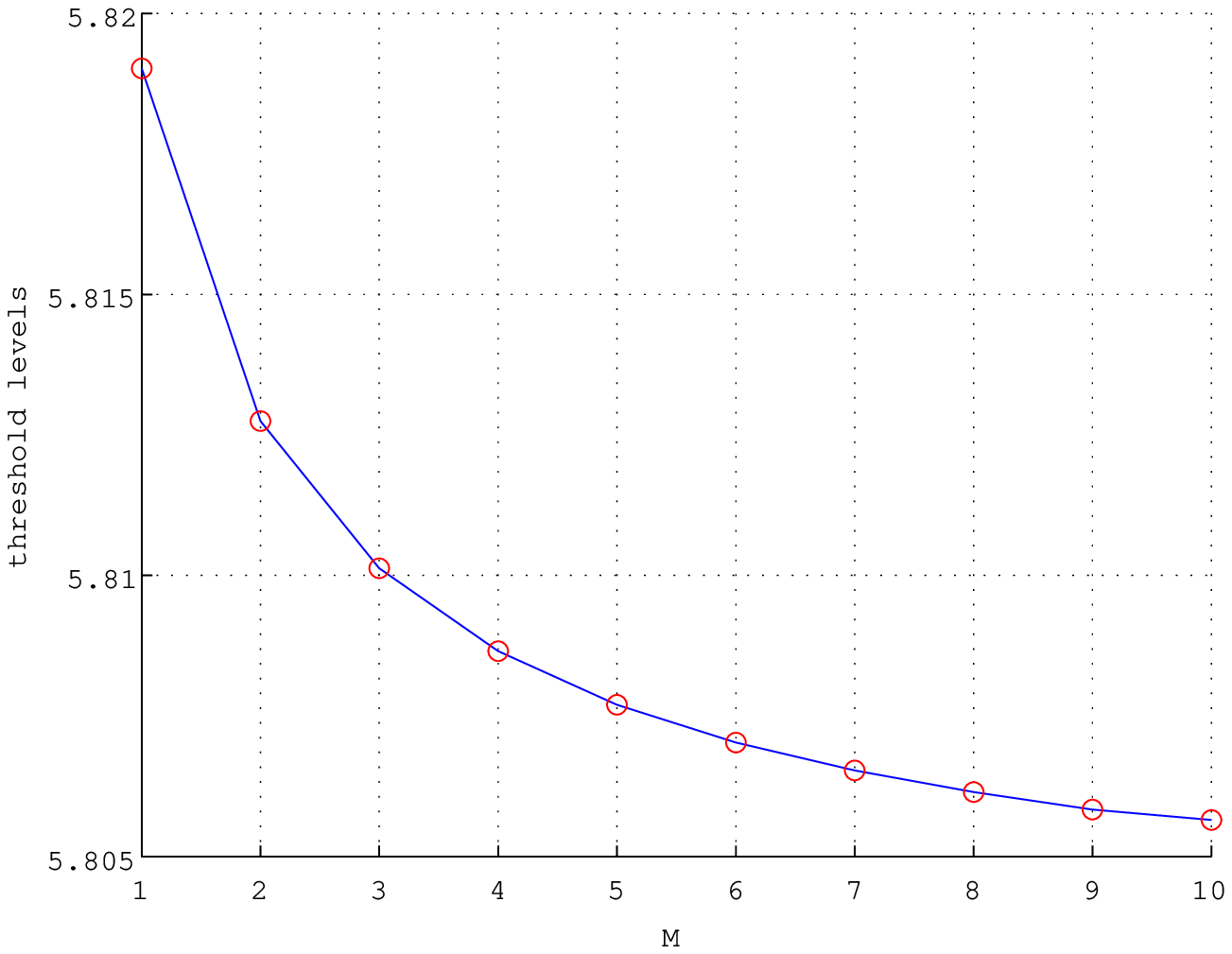}
\end{tabular}
\end{minipage}
\caption{Dependence of the thresholds  on $N$ and $M$  for \textbf{Case 3} with $\gamma = 0.1$. {The left panel plots $\widetilde{a}^*_1, \ldots, \widetilde{a}^*_4$ for fixed $M=1,\ldots, 4$.  The right panel plots  the threshold  $\widetilde{a}^*_2$ over  $M=1, \ldots, 10$.  }} \label{figure_convergence_threshold}
\end{center}
\end{figure}

\begin{figure}[h!]
\begin{center}
\begin{minipage}{1.0\textwidth}
\centering
\begin{tabular}{c}
 \includegraphics[scale=0.58]{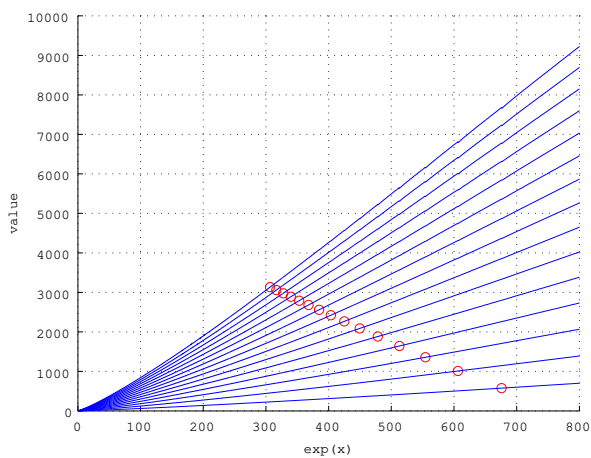}
\end{tabular}
\end{minipage}
\caption{Value functions and optimal thresholds {(marked by circles)} for \textbf{Case 3} with $N=1, \ldots, 15$ {when $M=1$ and} $\gamma = 0.05$.} \label{figure_convergence}
\end{center}
\end{figure}

\subsection{Limitations} Recall from formula \eqref{u_form} for  recovering the  function $u^{(n,m)}$ from the parameter set $\Gamma = (A,B,C,D,E)$.  In particular  the coefficients $D$'s are multiplied by $\exp(\Phi(\alpha + M/\delta) x) x^h$, so this term  tends to become very  large near $\widetilde{a}_1^*$, even though  $D$'s are zero above $\widetilde{a}_1^*$.  From our numerical tests,   it can take value up to the order of $10^{50}$ while the values of $D$'s tend to remain small.   Recall that $p := \alpha + M/\delta$ increases in $M$ and so does $\Phi(p) = \Phi(\alpha + M/\delta)$.   In addition the maximum value of $h$ (the counting index in \eqref{u_form}) increases as $M$ and $N$ increase.   As a result,  the computation can break down when the Erlang shape parameter $M$ and/or  the number of exercises $N$ are  large. MATLAB or other softwares with double precision cannot handle the computation involving these large numbers.

\begin{figure}[htbp]
\begin{center}
 \includegraphics[scale=0.57]{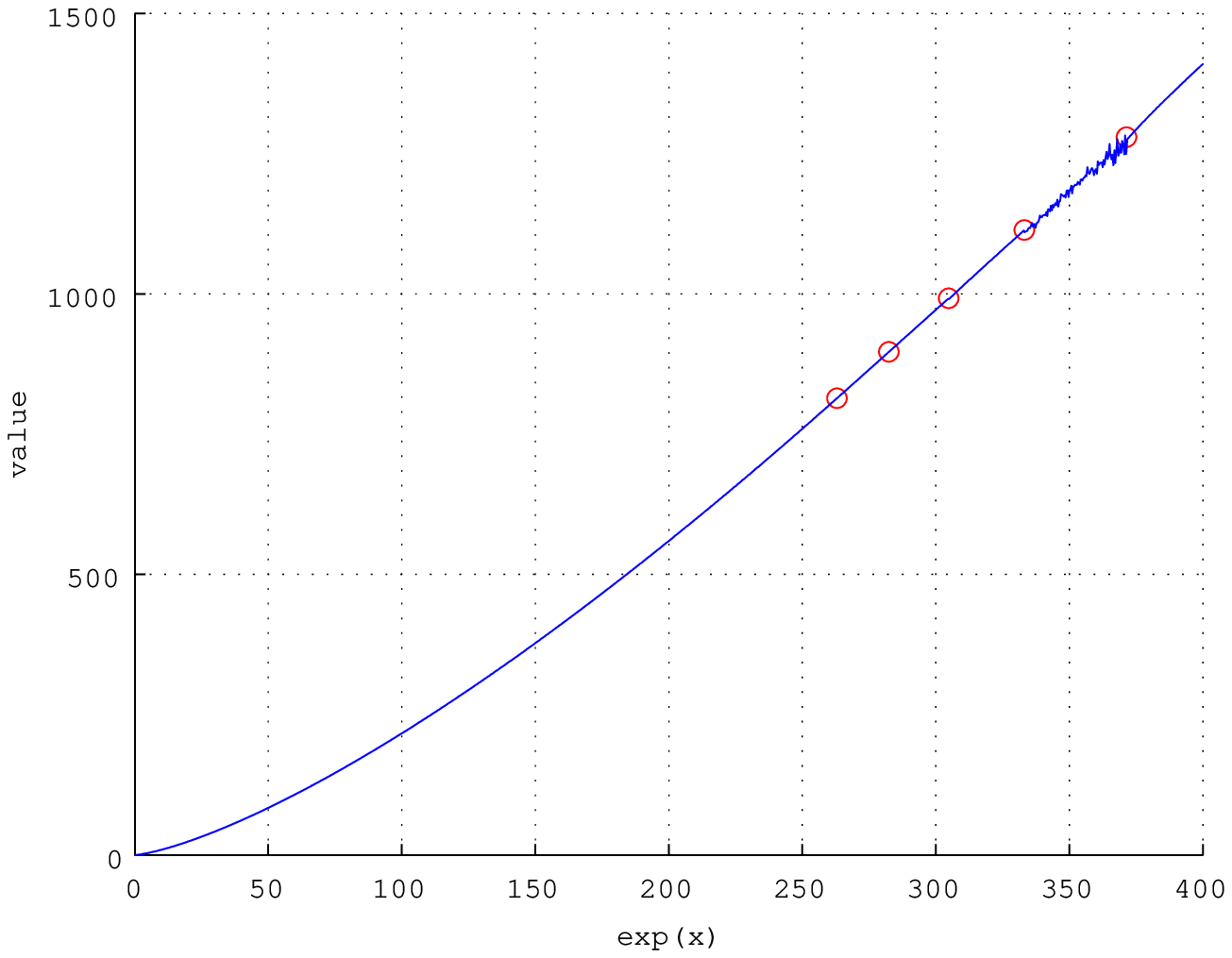}  \includegraphics[scale=0.57]{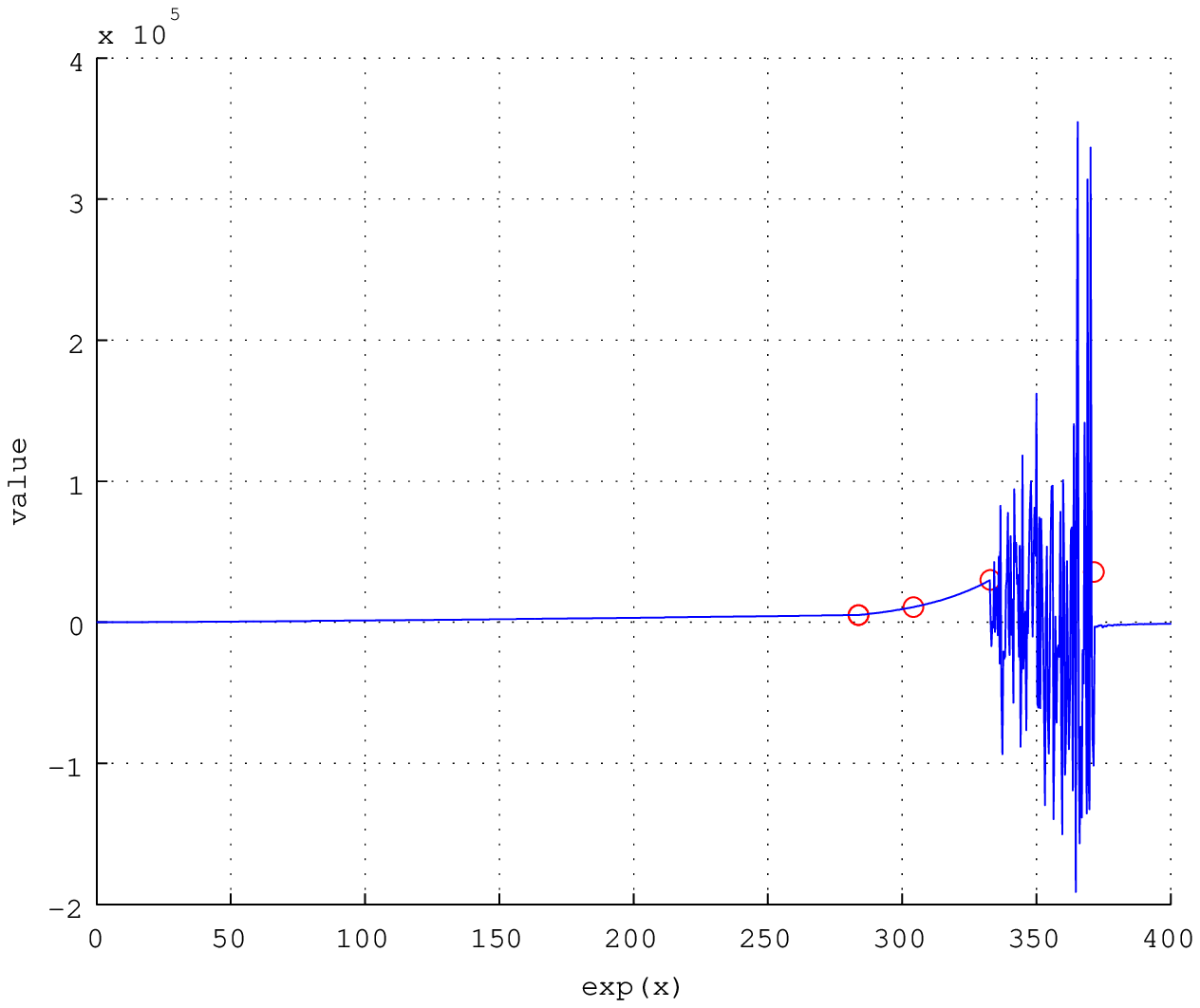}
\caption{Limitations: the value functions computed by MATLAB for \textbf{Case 3}  with $\gamma=0.1$ for $M=$4 (left), 5 (right). } \label{figure_limitation}
\end{center}
\end{figure}

In Figure \ref{figure_limitation}, we plot the function $\widetilde{v}^{(N)}$ computed by MATLAB for \textbf{Case 3} with $\gamma = 0.1$ and $N = 5$ for {$M=4$ (left) and $M=5$ (right)}.  While the parameters in  $\Gamma$ can be computed instantly and  do not explode,  numerical  imprecision   in computing the value function may  arise when  small parameters  are   multiplied   by very large numbers and summed up.
Indeed, with  {$M=4$}, discontinuities appear between $\widetilde{a}_2^*$ and $\widetilde{a}_1^*$, and with  {$M=5$} the error becomes visibly clear, yielding an inaccurate value function and  threshold levels $\widetilde{a}^*$.  This is consistent with the observation given in \cite{Kleinert_Schaik} (that deal with an American put option), where their randomization algorithm requires more than double precision.   This issue can potentially be resolved by setting the machine epsilon as in \cite{Kleinert_Schaik} so as to increase precision.  However, this is beyond the scope of our paper, because it requires special skills in computer science and it is our aim to evaluate the performance that can be achieved in a usual computing environment. Even given this potential limitation,   the analytic formula is useful in its own right as  it reveals the mathematical structure of the solution to the optimal multiple stopping problem.

 This observation also highlights the potential  trade-off between  selecting large  values of $M$ and $N$ given  limits on machine precision.
Nevertheless, we have seen from Tables \ref{table_randomization_simulation_002} and \ref{table_randomization_simulation_01}  that the approximation remains stable  for different small values of $M$. In summary, the analytic formula \eqref{u_form} is   very useful and tractable  for solving the optimal multiple stopping problem, as compared to the simulation approach.

\appendix

\section{Proof of Proposition \ref{proposition_recursion} and the updating formula}

Fix $n \geq 1$ and $0 \leq m \leq M-1$ and suppose \eqref{u_form}, \eqref{parameter_set} and \eqref{hypothesis_corner} hold.   We shall show that the identity \eqref{M_x_temporary} can be written as \eqref{updating_formula_equation} where the parameter set $\Gamma_{n,m+1}$
 is given by \eqref{update_formula} below.

First,  for $\alpha < \beta$, straightforward integration gives the following expression for $\varpi_{l}^{(n,m)}$ as in  \eqref{def_varpi_l_n_m}.

(1) When $q = \Phi(p)$,
\begin{align}
\begin{split}
\varpi_{l}^{(n,m)}(s,t, \Phi(p)) &= A^{(n,m,l)} \frac {e^{-\Phi(p) s}-e^{-\Phi(p) t}}{\Phi(p)}
+  B^{(n,m,l)}   \frac {e^{-(\Phi(p)-1) s}-e^{-(\Phi(p)-1) t}}{\Phi(p)-1} \\
&+ \sum_{j \in \mathcal{I}_p}  \sum_{h = 0}^{I_{n,m}}  e^{-(\Phi(p)+\xi_{j,p}) s} s^h \sum_{g=h}^{I_{n,m}} C_{j,g}^{(n,m,l)} \frac {g!} {h! (\Phi(p)+\xi_{j,p})^{g+1-h}}   \\ &- \sum_{j \in \mathcal{I}_p} \sum_{h = 0}^{I_{n,m}} e^{-(\Phi(p)+\xi_{j,p}) t} t^h \sum_{g=h}^{I_{n,m}} C_{j,g}^{(n,m,l)} \frac {g!} {h! (\Phi(p)+\xi_{j,p})^{g+1-h}}  \\
&+ \sum_{h=0}^{I_{n,m}} \frac{D_h^{(n,m,l)}}{h+1} (t^{h+1}-s^{h+1}) + E^{(n,m,l)}\frac {e^{-( \Phi(p)-\Phi(\alpha)) s}-e^{-( \Phi(p)-\Phi(\alpha)) t}}{ \Phi(p)-\Phi(\alpha)}.
\end{split}
\end{align}
In particular, for $t = \infty$, if $\Phi(p) > 1$ and $D^{(n,m,l)}=E^{(n,m,l)} = 0$
\begin{align}
\begin{split}
\varpi_{l}^{(n,m)}(s, \infty, \Phi(p)) &= A^{(n,m,l)} \frac {e^{-\Phi(p) s}}{\Phi(p)}
+  B^{(n,m,l)}   \frac {e^{-(\Phi(p)-1) s}}{\Phi(p)-1} \\
&+ \sum_{j \in \mathcal{I}_p}  \sum_{h = 0}^{I_{n,m}}  e^{-(\Phi(p)+\xi_{j,p}) s} s^h \sum_{g=h}^{I_{n,m}} C_{j,g}^{(n,m,l)} \frac {g!} {h! (\Phi(p)+\xi_{j,p})^{g+1-h}}.
\end{split}
\end{align}
(2) When $q = - \xi_{i,p}$, $\varpi_l^{(n,m)}(s,t, -\xi_{i,p})$ equals
\begin{align}
\begin{split}
 &A^{(n,m,l)} \frac {e^{\xi_{i,p} t}-e^{\xi_{i,p} s}}{\xi_{i,p}}
+  B^{(n,m,l)}   \frac {e^{(\xi_{i,p}+1) t}-e^{(\xi_{i,p}+1) s}}{\xi_{i,p}+1} +  \sum_{g=0}^{I_{n,m}} \frac {C_{i,g}^{(n,m,l)}} {g+1} (t^{g+1} - s^{g+1})  \\ &+   \sum_{j \in \mathcal{I}_p \backslash \{i\}}   \sum_{h= 0}^{I_{n,m}}  e^{(\xi_{i,p}-\xi_{j,p}) s} s^h  \sum_{g=h}^{I_{n,m}}    C_{j,g}^{(n,m,l)} \frac {g!} {h! (\xi_{j,p}-\xi_{i,p})^{g+1-h}} \\
&-  \sum_{j \in \mathcal{I}_p \backslash \{i\}}   \sum_{h = 0}^{I_{n,m}} e^{(\xi_{i,p}-\xi_{j,p}) t} t^h   \sum_{g=h}^{I_{n,m}}  C_{j,g}^{(n,m,l)} \frac {g!} {h! (\xi_{j,p}-\xi_{i,p})^{g+1-h}} \\
&+ \sum_{j = 0}^{I_{n,m}}  \Big[ e^{(\xi_{i,p}+\Phi(p)) s} s^j -   e^{(\xi_{i,p}+\Phi(p)) t} t^j \Big] \sum_{h=j}^{I_{n,m}} D_h^{(n,m,l)}  \frac {h!} {j! (-(\xi_{i,p}+\Phi(p)))^{h+1-j}} \\
&+ E^{(n,m,l)}\frac {e^{( \xi_{i,p}+\Phi(\alpha)) t}-e^{( \xi_{i,p}+\Phi(\alpha)) s}}{ \xi_{i,p}+\Phi(\alpha)}.
\end{split}
\end{align}
In particular, if $s = -\infty$ and $A^{(n,m,l)}=B^{(n,m,l)}=C^{(n,m,l)}=0$,
\begin{align}
\varpi_l^{(n,m)}(-\infty, t, -\xi_{i,p})  = -\sum_{j = 0}^{I_{n,m}}    e^{(\xi_{i,p}+\Phi(p)) t} t^j  \sum_{h=j}^{I_{n,m}} D_h^{(n,m,l)}  \frac {h!} {j! (-(\xi_{i,p}+\Phi(p)))^{h+1-j}} + E^{(n,m,l)}\frac {e^{( \xi_{i,p}+\Phi(\alpha)) t}}{ \xi_{i,p}+\Phi(\alpha)}.
\end{align}

By letting $s = x$ and $t = \widetilde{a}^*_{L-1}$ and multiplying by $\exp(\Phi(p)x)$, we obtain
\begin{multline} \label{varpi_expression_1}
e^{\Phi(p) x} \varpi_{L}(x,\widetilde{a}_{L-1}^*, \Phi(p)) = \hat{A}^{(n,m,L)} + \hat{B}^{(n,m,L)} e^x  \\ +  \sum_{i \in \mathcal{I}_p} \sum_{h=0}^{I_{n,m}+1} ( \hat{C}^{(n,m,L)}_{i,h} e^{- \xi_{i,p} x} x^h) +  \sum_{h=0}^{I_{n,m}+1} (\hat{D}^{(n,m,L)}_h e^{\Phi(p) x} x^h) + \hat{E}^{(n,m,L)} e^{\Phi(\alpha) x},
\end{multline}
where
\begin{align}
\begin{split}
\hat{A}^{(n,m,L)} &:=  {A^{(n,m,L)}} / {\Phi(p)}, \; \hat{B}^{(n,m,L)}:=  {B^{(n,m,L)}} /{(\Phi(p)-1)}, \\ \hat{C}^{(n,m,L)}_{i,h} &:=  \sum_{g=h}^{I_{n,m}} C_{i,g}^{(n,m,L)} \frac {g!} {h! (\Phi(p)+\xi_{i,p})^{g+1-h}},    \quad  0 \leq h \leq I_{n,m}, \quad  \hat{C}^{(n,m,L)}_{i,I_{n,m}+1} := 0,  \quad i \in \mathcal{I}_p,
\end{split}
\end{align}
and
\begin{align}
\begin{split}
\hat{D}^{(n,m,L)}_0 &:= 1_{\{L \geq 2\}}\Big[  \frac {-  A^{(n,m,L)}  e^{-\Phi(p) \widetilde{a}_{L-1}^*}}{\Phi(p)}  +  B^{(n,m,L)}   \frac {- e^{-(\Phi(p)-1) \widetilde{a}_{L-1}^*}}{\Phi(p)-1} \\ &\qquad - \sum_{j \in \mathcal{I}_p} \sum_{h = 0}^{I_{n,m}} e^{-(\Phi(p)+\xi_{j,p}) \widetilde{a}_{L-1}^*} (\widetilde{a}_{L-1}^*)^h \sum_{g=h}^{I_{n,m}} C_{j,g}^{(n,m,L)} \frac {g!} {h! (\Phi(p)+\xi_{j,p})^{g+1-h}} \\
&\qquad + \sum_{h=0}^{I_{n,m}} \frac{D_h^{(n,m,L)}}{h+1} (\widetilde{a}_{L-1}^*)^{h+1}  - \frac {E^{(n,m,L)}} { \Phi(p)-\Phi(\alpha)}  e^{-( \Phi(p)-\Phi(\alpha)) \widetilde{a}_{L-1}^*}
\Big], \\
\hat{D}^{(n,m,L)}_h &:= - 1_{\{L \geq 2\}} {D_{h-1}^{(n,m,L)}}/{h}, \quad  1 \leq h \leq I_{n,m}+1, \\
\hat{E}^{(n,m,L)} &:= 1_{\{L \geq 2\}} {E^{(n,m,L)}}/{ (\Phi(p)-\Phi(\alpha))}.
\end{split}
\end{align}

Similarly, we have
\begin{multline} \label{varpi_expression_2}
\sum_{i \in \mathcal{I}_p} \kappa_{i,p}   e^{-\xi_{i,p} x} \varpi_{L}^{(n,m)}(\widetilde{a}_{L}^*,x, -\xi_{i,p}) = \check{A}^{(n,m,L)}  +e^x \check{B}^{(n,m,L)} + \sum_{i \in \mathcal{I}_p}  \sum_{h=0}^{I_{n,m}+1} e^{- \xi_{i,p} x} x^h  \check{C}^{(n,m, L)}_{i,h}\\
+ \sum_{h=0}^{I_{n,m}+1}  e^{\Phi(p) x} x^h \check{D}^{(n,m,L)}_h  + e^{\Phi(\alpha) x} \check{E}^{(n,m, L, i)},
\end{multline}
where $\check{A}^{(n,m,L)} := \sum_{i \in \mathcal{I}_p}   \kappa_{i,p}{A^{(n,m,L)}} /{\xi_{i,p}}$, $\check{B}^{(n,m,L)} := \sum_{i \in \mathcal{I}_p}   \kappa_{i,p} {B^{(n,m,L)}}/{(\xi_{i,p}+1)}$, \\ $\check{C}^{(n,m,L)}_{j,h} := \sum_{i \in \mathcal{I}_p}   \kappa_{i,p} \check{C}^{(n,m, L,i)}_{j,h}$,
$\check{D}^{(n,m,L)}_h := \sum_{i \in \mathcal{I}_p}   \kappa_{i,p} \check{D}^{(n,m, L,i)}_h$,
$\check{E}^{(n,m, L)} := \sum_{i \in \mathcal{I}_p}   \kappa_{i,p} { E^{(n,m,L)}} /{ (\xi_{i,p}+\Phi(\alpha))}$,
with, for all $i \in \mathcal{I}_p$,
\begin{align}
\begin{split}
%\check{A}^{(n,m,L,i)} &:= {A^{(n,m,L)}} /{\xi_{i,p}}, \; \check{B}^{(n,m,L,i)} :=  {B^{(n,m,L)}}/{(\xi_{i,p}+1)}, \\
\check{C}^{(n,m,L,i)}_{i,0} &:= - A^{(n,m,L)} \frac { e^{\xi_{i,p} \widetilde{a}_{L}^*}}{\xi_{i,p}}
-  B^{(n,m,L)}   \frac {e^{(\xi_{i,p}+1) \widetilde{a}_{L}^*}}{\xi_{i,p}+1} - \sum_{g=0}^{I_{n,m}} \frac {C_{i,g}^{(n,m,L)}} {g+1} (\widetilde{a}_{L}^*)^{g+1}  \\
 &\qquad + \sum_{j \in \mathcal{I}_p \backslash \{i\}}   \sum_{h= 0}^{I_{n,m}} (\widetilde{a}_{L}^*)^h   e^{(\xi_{i,p}-\xi_{j,p}) \widetilde{a}_{L}^*}  \sum_{g=h}^{I_{n,m}}    C_{j,g}^{(n,m,L)} \frac {g!} {h! (\xi_{j,p}-\xi_{i,p})^{g+1-h}} \\
&\qquad + \sum_{j = 0}^{I_{n,m}}  \Big[e^{(\xi_{i,p}+\Phi(p)) \widetilde{a}_{L}^*} (\widetilde{a}_{L}^*)^j \Big] \sum_{h=j}^{I_{n,m}} D_h^{(n,m,L)}  \frac {h!} {j! (-(\xi_{i,p}+\Phi(p)))^{h+1-j}} \\
&\qquad - E^{(n,m,L)}\frac { e^{( \xi_{i,p}+\Phi(\alpha)) \widetilde{a}_{L}^*}}{ \xi_{i,p}+\Phi(\alpha)}, \\
\check{C}^{(n,m,L,i)}_{i,h} &:=   {C_{i,h-1}^{(n,m,L)}} / {h}, \quad 1 \leq h \leq I_{n,m}+1,   \\
\check{C}^{(n,m,L,i)}_{j,h} &:= - \sum_{g=h}^{I_{n,m}}  C_{j,g}^{(n,m,L)} \frac {g!} {h! (\xi_{j,p}-\xi_{i,p})^{g+1-h}}, \quad 0 \leq h \leq I_{n,m}, \quad \check{C}^{(n,m,L,i)}_{j, I_{n,m}+1} := 0, \quad   j \in \mathcal{I}_p \backslash \{ i \},  \\
\check{D}^{(n,m,L,i)}_h &:= -  \sum_{g=h}^{I_{n,m}} D_g^{(n,m,L)}  \frac {g!} {h! (-(\xi_{i,p}+\Phi(p)))^{g+1-h}}, \quad 0 \leq h \leq I_{n,m}, \quad \check{D}^{(n,m,L, i)}_{I_{n,m}+1} := 0.
%\check{E}^{(n,m,L,i)} &:=   { E^{(n,m,L)}} /{ (\xi_{i,p}+\Phi(\alpha))}.
\end{split}
\end{align}

Substituting \eqref{varpi_expression_1} and \eqref{varpi_expression_2} in \eqref{M_x_temporary}, Proposition \ref{proposition_recursion}   is satisfied by setting $I_{n,m+1}=I_{n,m}+1$ and
\begin{align} \label{update_formula}
\begin{split}
A^{(n,m+1,L)} &:=   1_{\{L \leq n\}} (\Phi'(p)\hat{A}^{(n,m,L)} +  \check{A}^{(n,m,L)}), \\
B^{(n,m+1,L)} &:= 1_{\{L \leq n\}} (\Phi'(p) \hat{B}^{(n,m,L)} +  \check{B}^{(n,m,L)}), \\
C^{(n,m+1,L)}_{i,0} &:= 1_{\{L \leq n\}} \big( \Phi'(p) \hat{C}_{i,0}^{(n,m,L)} +  \check{C}_{i,0}^{(n,m,L)} +   \kappa_{i,p}   \sum_{L+1 \leq l \leq n+1} \varpi_{l}^{(n,m)}(\widetilde{a}_{l}^*,\widetilde{a}_{l-1}^*, -\xi_{i,p}) \big), \quad i \in \mathcal{I}_p, \\
C^{(n,m+1,L)}_{i,h} &:= 1_{\{L \leq n\}} (\Phi'(p) \hat{C}_{i,h}^{(n,m,L)} + \check{C}_{i,h}^{(n,m,L)}), \quad 1 \leq h \leq I_{n,m+1}, \; i \in \mathcal{I}_p, \\
D^{(n,m+1,L)}_{0} &:= 1_{\{L \geq 2\}}  \big(\Phi'(p)\hat{D}_{0}^{(n,m,L)} + \check{D}_{0}^{(n,m,L)}+ \Phi'(p) \sum_{1 \leq l \leq L-1} \varpi_{l}^{(n,m)}(\widetilde{a}_l^*,\widetilde{a}_{l-1}^*, \Phi(p)) \big), \\
D^{(n,m+1,L)}_{h} &:= 1_{\{L \geq 2\}} \big(\Phi'(p) \hat{D}_{h}^{(n,m,L)} + \check{D}_{h}^{(n,m,L)}\big),  \quad 1 \leq h \leq I_{n,m+1}, \\
E^{(n,m+1,L)} &:= 1_{\{L \geq 2\}}  (\Phi'(p) \hat{E}^{(n,m,L)} + \check{E}^{(n,m,L)}).
\end{split}
\end{align}
This completes the proof.

\section*{Acknowledgements} The authors thank the two anonymous referees for their insightful comments.  The second author is  supported by MEXT KAKENHI Grant Number  26800092, the Inamori foundation research grant, and the Kansai University Subsidy for Supporting young Scholars 2014.

\bibliographystyle{abbrv}
 
\bibliography{ospbib2}

\end{document}